\documentclass[11pt]{amsart}
\usepackage[margin=1in]{geometry}
\usepackage{amsfonts, float}
\usepackage{amssymb}
\usepackage{amsmath}
\usepackage{bbm}

\usepackage{color}
\usepackage{cite}
\usepackage[pdftex]{hyperref}
\numberwithin{equation}{section}
\newcommand{\eq}[1]{(\ref{#1})}

\newcommand{\m}[1]{\mathbb{#1}}
\newcommand{\mc}[1]{\mathcal{#1}}
\newcommand{\mf}[1]{\mathfrak{#1}}

\newcommand{\rr}{\m{R}}
\newcommand{\Rl}{\rr}

\newcommand{\nn}{\m{N}}
\newcommand{\cc}{\m{C}}
\newcommand{\zz}{\m{Z}}
\newcommand{\Ir}{\zz}

\newcommand{\1}{\mathbbm{1}} 
\newcommand{\spec}[1]{\mathrm{sp}(#1)}

\newcommand{\cA}{\A}
\newcommand{\cP}{\mathcal{P}}
\newcommand{\cB}{\mathcal{B}}
\newcommand{\cK}{\mathcal{K}}
\newcommand{\cS}{\mathcal{S}}
\newcommand{\be}{\begin{equation}}
\newcommand{\ee}{\end{equation}}
\newcommand{\bea}{\begin{eqnarray}}
\newcommand{\eea}{\end{eqnarray}}
\newcommand{\beann}{\begin{eqnarray*}}
\newcommand{\eeann}{\end{eqnarray*}}
\newcommand{\idty}{\1}

\newcommand{\la}{\lambda}
\newcommand{\La}{\Lambda}
\newcommand{\sig}{\sigma}
\newcommand{\ep}{\varepsilon}

\newcommand{\om}{\omega}
\newcommand{\vtheta}{\vartheta}

\newcommand{\al}[1]{\begin{equation} #1 \end{equation}} 
\newcommand{\set}[1]{\left\{ #1 \right\} }
\newcommand{\ip}[1]{\langle #1 \rangle}
\newcommand{\norm}[1]{\Vert #1 \Vert }

\newcommand{\tr}[1]{\text{tr}(#1)}
\newcommand{\trr}{\text{tr}}

\renewcommand{\Re}{\mathrm{Re\,}}
\renewcommand{\Im}{\mathrm{Im\,}}

\newcommand{\vertiii}[1]{{\left\vert\kern-0.25ex\left\vert\kern-0.25ex\left\vert #1 
    \right\vert\kern-0.25ex\right\vert\kern-0.25ex\right\vert}} 

\newtheorem{theorem}{Theorem}[section]
\newtheorem{lemma}[theorem]{Lemma}
\newtheorem{proposition}[theorem]{Proposition}
\newtheorem{corollary}[theorem]{Corollary}

\newenvironment{definition}[1][Definition.]{\begin{trivlist}
\item[\hskip \labelsep {\bfseries #1}]}{\end{trivlist}}
\newenvironment{example}[1][Example.]{\begin{trivlist}
\item[\hskip \labelsep {\bfseries #1}]}{\end{trivlist}}

\newcommand{\pert}[1]{\widetilde{#1}}

\newcommand{\vphi}{\varphi}
\newcommand{\specc}{\text{sp}}

\newcommand{\floor}[1]{\lfloor #1 \rfloor}
\newcommand{\diam}[1]{\mathrm{diam}(#1)}

\newcommand{\A}{\mf{A}}
\newcommand{\Aloc}{\mf{A}_{\mathrm{loc}}}
\newcommand{\ALa}{\mf{A}_{\La}}

\newcommand{\Aut}{\text{Aut}}
\newcommand{\JW}[1]{\vartheta(#1)}
\newcommand{\Int}{\mathrm{Int}}

\title[Stability in One Dimension]{Stability of Gapped Ground State Phases\\ of Spins and Fermions in One Dimension}

\author[A. Moon]{Alvin Moon}
\address{Department of Mathematics\\
University of California, Davis\\
Davis, CA 95616, USA}
\email{asmoon@math.ucdavis.edu}

\author[B. Nachtergaele]{Bruno Nachtergaele}
\address{Department of Mathematics and Center for Quantum Mathematics and Physics\\
University of California, Davis\\
Davis, CA 95616, USA}
\email{bxn@math.ucdavis.edu}

\date{}
 
\begin{document}

\maketitle

\begin{abstract}
We investigate the persistence of spectral gaps of one-dimensional frustration free quantum lattice systems under weak perturbations and with open boundary conditions. Assuming the interactions of the system satisfy a form of local topological quantum order, we prove explicit lower bounds on the ground state spectral gap and higher gaps for spin and fermion chains. By adapting previous methods using the spectral flow, we analyze the bulk and edge dependence of lower bounds on spectral gaps.  
\end{abstract}

\begin{center}
{\em Dedicated to the memory of Ludwig Faddeev}
\end{center}

\section{Introduction} 

An important result in  the study of gapped ground state phases of quantum lattice systems (with or without topological order) is the stability of 
the spectral gap(s) under uniformly small extensive perturbations. The stability property implies that the gapped phases are full-dimensional 
regions in the space of Hamiltonians free of phase transitions \cite{bachmann:2012}.
In recent years, such results were obtained in increasing generality 
\cite{Kennedy:1992,Datta:1996,Yarotsky:2006,bravyi:2010,Michalakis:2013,szehr:2015,young:2016,hastings:2017}. Our goal here is to extend the existing results applicable in one 
dimension to Hamiltonians with so-called `open' boundary conditions, meaning that we consider systems defined on intervals $[a,b]\subset 
\Ir$ and not on a cycle $\Ir/(n\Ir)$. Specifically, this implies that the neighborhoods of the boundary points $a$ and $b$ may be treated differently 
than the bulk. There are physical and mathematical situations where one is naturally led to considering open boundary conditions. For example, in the 
series of recent works by Ogata \cite{ogata:2016a,ogata:2016,ogata:2017}, clarifying the crucial role of boundary states in the classification of quantum 
spin chains with matrix product ground states required the study of systems with open boundary conditions. Another situation of interest to us is the 
application of results for quantum spin chains to fermion models in one dimension by making use of the Jordan-Wigner transformation, which in the 
finite system set-up only works well with open boundary conditions. In this way, we obtain explicit bounds on the spectral gaps in the spectrum of 
perturbed spin and even fermion chains with one or more frustration free ground states that satisfy a local topological order condition. 
This complements previous results that prove stability of gapped fermion systems by other approaches 
\cite{nachtergaele:2016b, hastings:2017,roeck:2017}.

\section{Setting and Main Result}

\subsection{Notations}

Denote by $(\zz, |\cdot|)$ the metric graph of integers. Let $P_f(X)$ denote the finite subsets of $X\subset \zz$.  We will use $\La$ to refer exclusively to nonempty, finite intervals of the form $[a,b] = \set{ n\in \nn : a \leq n \leq b}$. Let $b_\La(x,n) = \set{ m \in \La: |x-m|\leq n}$ denote the restriction of a metric ball to the interval. For each $x\in \La$, denote by $r_x$ and $R_x$ the following distances to the boundary:
\begin{equation}
r_x = \min \set{x-a, b-x  },\quad R_x = \max \set{ x-a, b-x}
\label{eq:distances}\end{equation}
Although $r_x$ and $R_x$ depends on the interval $[a,b]$, we omit this dependence from the notation since we will always fix a finite 
volume $[a,b]$ throughout our arguments.

In the following, we will consider both spin systems and fermion systems on the one-dimensional lattice. Without difficulty, we could
also treat systems that include both types of degrees of freedom, but for simplicity of the notations, we will not consider such systems 
in this paper. It is also possible to consider inhomogeneous systems for which the number of spin or fermion states depends on the site.
In order to present the main ideas without overly burdensome notation, we will only consider homogeneous systems in the note.

The algebra of observables of the finite system in $\Lambda$, of either spins or fermions, will be denoted by $\mf{A}_\Lambda$. If we want
to specify that we are specifically considering spins or fermions, we will use the notation $\mf{A}_\Lambda^{s}$ or $\mf{A}_\Lambda^{f}$,
respectively. These algebras, and the associated Hilbert space they are represented on, are defined as follows. 

For spin systems, we have
\begin{equation*}
\mf{A}_\Lambda^{s} = M_d (\cc)^{\otimes |\La|} , \quad \mf{h}_\La = (\cc^d)^{\otimes |\La|},
\end{equation*}
where $d$ is the dimension of the Hilbert space of a single spin, i.e., $d=2S+1$. 

For fermions, $\ALa^f$ denotes the $C^*$-algebra generated by $\set{ a(x), a^*(x): x\in \La}$, the annihilation and creation operators
defining a representation of the Canonical Anticommutation Relations (CAR) on the  antisymmetric Fock space $\mf{F}_\La=\mf{F}(\ell^2(\La))$. The 
dimension of $\mf{F}_\La$ is $2^{|\La|}$ and $\ALa^f$ is $*$-isomorphic to the matrix algebra $M_{2^{|\La|}}(\cc)$. 

Given an exhaustive net of CAR or spin algebras $\set{\ALa: \La \in P_f(\zz)}$, the inductive limit $\A_\zz$, the $d^\infty$ UHF algebra, is obtained by norm completion: 
	\begin{equation*}
		\begin{split}
\A_\zz = \overline{ \bigcup _{\La \in P_f(\zz)} \A_\La}. 
		\end{split}
	\end{equation*}
This algebra is often referred to as the \textit{quasi-local} algebra, and $\Aloc = \bigcup \ALa$ as the \textit{local} algebra. 


Define by $N_X=\sum_{x\in X} a^*(x) a(x)$ the number operator for $X \in P_f(\zz)$, and define the \textit{parity automorphism} by: 
	\begin{equation}\label{eq:parity-symmetry}
		\begin{split}
\rho_\La(A) = \exp(i\pi N_\La) A \exp(i\pi N_\La)
		\end{split}
	\end{equation}
Say that $A \in \ALa^f$ is \textit{even} if $\rho_\La(A)=A$ and odd if $\rho_\La(A)=-A$. The observable $A$ is even if and only if it commutes with the local symmetry operator $\exp(i\pi N_\La)$, which is if and only if $A$ is the sum of even monomials in the generating set $\set{a(x),a^*(x):x\in \La}$. Unlike the odd observables, the even observables form a $*$-subalgebra of $\ALa^f$, which we denote by $\ALa^+$.  

\subsection{Assumptions}\label{sec:assumptions}


Let $I$ be a subinterval of $\zz$, not necessarily finite. An \textit{interaction on} $I$ is a function $\Phi: P_f(I) \to \Aloc$ such that $\Phi(X) = \Phi(X)^* \in \A_{X}$ for all $X \in P_f(I)$. The corresponding \textit{local Hamiltonian} of the finite system on $\La \subset I$ is $H_\La = \sum _{X \subset \La} \Phi(X)$.  Say $\Phi$ is \textit{non-negative} if $\Phi(X) \geq 0$ for all $X\in P_f(I)$. Say $\Phi$ is an \textit{even} interaction of the CAR algebra if $\Phi(X) \in \mf{A}_X^+$. 

The interactions in our perturbative set-up will satisfy the following assumptions. First, let $\eta: P_f(\zz) \to \Aloc$ be a non-negative interaction with distinguished local Hamiltonians $H_\La$. We will refer to $\eta$ as the \textit{unperturbed interaction}. We assume $\eta$ has the following properties: 
\begin{enumerate}
\item \textbf{Finite range}: there exists $R>0$ such that $\diam{X}>R$ implies $\eta(X)=0$. 
\item \textbf{Uniformly bounded}: there exists $M>0$ such that for all $X \in P_f(\zz)$, $\norm{\eta(X)}<M$. 
\item \textbf{Frustration free}: for all intervals $\La \in P_f(\zz)$, $\ker(H_\La) \supsetneq \set{0}$. 
\item \textbf{Uniformly locally gapped}: There exists $\gamma_0>0$ such that for all intervals $[a,b] \in P_f(\zz)$, with $b-a\geq R$, $\gamma_0$
is lower bound a for non-zero eigenvalues of $H_{[a,b]}$. 
\item \textbf{\textit{Local topological quantum order}} (\textbf{LTQO}) of the ground state projectors. 
\end{enumerate} 
The concept of LTQO was introduced in \cite{bravyi:2010}. We will need to adapt the definition to take into account 
parity and boundary conditions, which we do in the next section. 

Next, we consider the perturbations. To allow edge effects, we will consider perturbations given in terms of a family of interactions on intervals. For each $\La$, let $\Phi^\La : P_f(\La) \to \Aloc$ be an interaction on the interval, and denote by $[\Phi]$ the collection of these perturbative interactions:
	\begin{equation}\label{eq:pert-class}
		\begin{split}
[\Phi] = \set{ \Phi^\La : \La \in P_f(\zz)}
		\end{split}
	\end{equation}
The perturbed Hamiltonians have the form:
	\begin{equation}\label{eq:deform}
		\begin{split}
H_\La(\ep) = \sum _{X \subset \La} \eta(X)+\ep \Phi^\La(X), ~\ep \in [0,1]
		\end{split}
	\end{equation}
and while the Hamiltonians depend on the interval $\La$, lower bounds on gaps in the spectrum will be uniform in the volume.

\medskip 

Our main assumption on the interactions $\Phi^\La$ in $[\Phi]$ is that $\Phi^\La(X)$ decays rapidly with the diameter of $X$. To make this precise, we use $\mc{F}$-functions and provide explicit bounds in terms of the $\mc{F}$-norm. The definition and properties of 
$\mc{F}$-functions and $\mc{F}$-norm can be found in the Appendix. In our argument, we will use functions of the form: 
	\begin{equation}\label{eq:base-f}
		\begin{split}
F(x) & = e^{-h(x)} F^b(x) \\ 		
F^b(x) & =  \frac{L}{(1+cx)^{\kappa}}
		\end{split}
	\end{equation}
where $\kappa>2$ and $L,c>0$. The function $h:[0,\infty) \to [0,\infty)$ is a monotone increasing, subadditive weight function. At times, it will be necessary to precompose $F$ with a transformation $\tau: [0,\infty) \to \rr$, and so we will take as convention $F\circ \tau(x)=F(0)$ for $\tau(x)<0$. We will denote by $\norm{\cdot}_F$ the extended norm (\ref{eq:fnorm}) induced by $F$. 
 
Using $\mc{F}$-function terminology, we assume for the perturbations:
\begin{enumerate}
\item \textbf{Fast decay}: there exists an $\mc{F}$-function $F (r) = e^{-h_\Phi(r)} \frac{L}{(1+cx)^\kappa}$, for $L,c>0$ and $\kappa>2$, such that 
$\sup _\La \norm{ \Phi^\La}_F < \infty$.
\item \textbf{Metric ball support}: for all $\La$, $\Phi^\La(X) \not = 0$ implies $X = b_\La(z,n)$ for some $z\in \La$ and $n\in \nn$. 
\end{enumerate} 
The assumption that $\Phi^\La$ is supported on metric balls is not restrictive, since a finite-volume Hamiltonian of any fast-decaying interaction can be rewritten as the finite-volume Hamiltonian of a balled interaction with comparable decay (c.f. the appendix of \cite{young:2016}). 

\subsection{Local topological quantum order}
%

Consider the unperturbed interaction $\eta$ and its local Hamiltonians. Denote by $P_X$ the orthogonal projection onto $\ker{(H_X)}$, and define the state:
\begin{equation*}
\begin{split}
\omega_\La (A) = \frac{1}{\tr{P_\La}} \tr{P_\La A}, ~ A \in \A_\La
\end{split}
\end{equation*}

\begin{definition}\label{def:LTQO}
The unperturbed interaction $\eta$ satisfies \textit{local topological quantum order} if there exists a monotone function 
$\Omega: [0,\infty) \to [0,\infty)$, decreasing to $0$, such that for all $x\in \La$ and $n,k\in \nn$ satisfying $0\leq k \leq r_x$ and $k \leq n \leq R_x$, the following bound holds: 
\begin{equation}\label{eq:LTQO-bound}
\begin{split}
\forall A \in \A_{b_\La(x,k)}: \norm{ P_{b_\La(x,n)} (A - \omega_\La(A)) P_{b_\La(x,n)} } \leq \Omega(z_x(n)-k) \norm{A} 
\end{split}
\end{equation}  
where $z_x: \nn \to \nn$ is the cut-off function defined in terms of distance to the boundary of $\La$ (\ref{eq:distances}):
\begin{equation}\label{eq:cut-off}
\begin{split}
z_x(m) = \bigg{\{} \begin{array}{l l}m & \text{ if } m \leq r_x \\ r_x & \text{ else} \end{array}
\end{split} 
\end{equation}
If $\eta: P_f(\zz) \to \Aloc^f$ is an even interaction and (\ref{eq:LTQO-bound}) holds for the restricted class of observables $A \in \A^+_{b_\La(x,k)}$, then we will say $\eta$ has \textit{$\zz_2$-LTQO}. 
\end{definition}

For example, the AKLT interaction with either periodic or open boundary conditions has LTQO with $\Omega(r) = (1/3)^{r}$. The interaction defined in (\ref{eq:orbitals}) has $\zz_2$-LTQO with $\Omega(r) = 0$ for $r$ greater than a cut-off $D>0$ defined by the interaction parameters, and $\Omega(x)=2$ otherwise (Proposition \ref{prop:orbitals-LTQO}).   

\subsection{The main result} 

For any finite interval $\La$, we consider the local Hamiltonian $H_\La(\ep)$ given in \eq{eq:deform}.
There exist continuous functions $\la_1,\ldots,\la_N: [0,1] \to \rr$ such that for all $\ep \in [0,1]$, $\{\la_1(\ep),\ldots,\lambda_N(\ep)\}$ are the eigenvalues 
of $H_\La(\ep)$. We partition $\spec{H_\La(\ep)}$ into two disjoint regions, an upper and a lower part of the spectrum, and call the minimum distance between these two sets the {\em spectral gap above the ground state} or {\em the spectral gap}:  
\al{ \specc_{0,\La}(\ep) = \set{ \la_i(\ep) : \la_i(0)=0 } \hspace{20mm} \specc_{1,\La}(\ep) = \set{ \la_j(\ep): \la_j(0)> 0 }.\label{def_sp0}}
\al{ \gamma(H_\La(\ep)) = \min \set{ \la - \mu : \la \in \specc_{1,\La}(\ep),~\mu \in \specc_{0,\La}(\ep) }}

For a class of sufficiently small perturbations, the main result of this paper establishes a lower bound for the size of the spectral gap which does not depend 
on $\La$, under the assumptions that $\eta$ has LTQO, the interactions in $[\Phi]$, from (\ref{eq:pert-class}), decay sufficiently fast and, in the case of fermions, that the interactions are even. The spectrum may have other gaps which can be defined similarly in terms of eigenvalue splitting, and we also prove an estimate showing how these gaps persist under weak perturbations. To state these results, we define several constants that characterize the effect of the perturbation and the presence of edge effects. 

The effect of perturbations near the boundary of $\Lambda$ is, in general, different and stronger than far away from the boundary. As a consequence, 
our stability result for open chains features a distance parameter $D \geq 0$, in terms of which we distinguish sites near and far away from the boundary. In Section 3, we write each $\Phi^\La$  as the sum of an interaction $\Phi^D(\La)$, with local Hamiltonian $\Phi^D_\La$ supported at the $D$-boundary, and a bulk interaction $\Phi^\Int(\La)$. Define the following two finite constants quantifying the strength of the bulk and edge perturbations, respectively: 
	\begin{gather*}
M_\Int = \sup_\La \set{ \norm{ \Phi^\Int(\La)}_{F}: \diam{\La} > \max\set{2D, R}} \\ 
M_D = \sup _\La \set{ \norm{\Phi_\La^D}: \diam{\La} > \max\set{2D,R}}
	\end{gather*}
Then, for constant:
	\begin{equation*}
		\begin{split}		
m  = \bigg{(} \sum _{|n|\geq 3} 20 C ( 3|n| + 2) \bigg{[} \Omega\bigg{(} \frac{|n|-1}{2} \bigg{)}^{1/2} &+ F_0 \bigg{(} \frac{|n|-3}{2} \bigg{)} \bigg{]}  \\
\hspace{20mm}
& + C\bigg{(} \sum _{n\in \zz} \Omega\bigg{(}\frac{|n|}{2} \bigg{)} + 2 F_0 \bigg{(} \floor{ \frac{|n|}{2}} \bigg{)}+8 \bigg{)}  \bigg{)} ( \norm{\eta}_{F} + M_\Int )
		\end{split}
	\end{equation*}
where $F_0(x) = F^b(x/18 - R -3/2)$, we are able to prove the following theorem. 

\noindent\textbf{Theorem 3.11}(Ground state gap stability for spin chains)\textbf{.} \textit{Suppose $\eta: P_f(\zz) \to \Aloc^s$ has LTQO with $\Omega(n) \leq n^{-\nu}$, for $\nu>4$, and there exist $K>0$, $s\in (0,1]$ such that $h_\Phi$ satisfies $h_\Phi(r) \geq Kr^s$. Then there exists $\ep(\gamma_0)>0$ such that $0\leq \ep < \ep(\gamma_0)$ and $\diam{\La}>\max\set{2D,R}$ imply: 
	\begin{equation*}
		\begin{split}
\gamma(H_\La(\ep)) \geq \gamma_0 - (m+2M_D) \ep >0
		\end{split}
	\end{equation*}
The constant $ \ep(\gamma_0)$ can be taken as:
\begin{equation*}
	\begin{split}
\ep(\gamma_0) & = \min \set{1, \frac{\gamma_0}{m + 2 M_D}}
	\end{split}
	\end{equation*}
}
As a consequence, if we assume $\eta: P_f(\zz) \to \Aloc^+$ has $\zz_2$-LTQO, and $\Omega$ and $\Phi^\La : P_f(\La) \to \ALa^+$ have the same decay assumptions as in Theorem \ref{thm:spin-gap-stability}, we are also able to prove: 
	
\noindent\textbf{Theorem \ref{thm:fermion-gap-stability}.} 	\textit{There exist $\ep'(\gamma_0)>0$ and constant $m_D'$ such that $0\leq \ep < \ep'(\gamma_0)$ and $\diam{\La}>\max\set{2D,R}$ implies: 
	\begin{equation*}
		\begin{split}
\gamma(H_\La(\ep)) \geq \gamma_0 - m_D'\ep > 0
		\end{split}
	\end{equation*}
The constants $m_D'$ and $\ep'(\gamma_0)$ can be explicitly determined by the constants $m, M_D$ and $\ep(\gamma_0)$ }


The proofs of Theorems \ref{thm:spin-gap-stability} and \ref{thm:fermion-gap-stability} rely on a relative form bound argument. We remark that the proof will depend strongly on the fact that
the size of the boundary of $\Lambda$ can be bounded independently of the size of $\Lambda$ itself. This is special about one-dimensional systems.
The stability of the gap in higher dimensions requires a careful analysis of the locality of perturbations \cite{young:inprep} and more complicated assumptions. 

Additionally, due to the relative form bound, the hypotheses for a stable ground state spectral gap also imply general stability of the spectrum. Precisely, we prove the following  statement about the persistence of higher spectral gaps. In the statement, $J_1,J_2,J_3$ refer to equations (\ref{eq:Jconstants}) and (\ref{eq:third-J-constant}).

\noindent\textbf{Proposition 3.12.} \textit{ 
Let $T,\gamma>0$, and denote $\mathrm{res}(H_\La) = \cc \setminus \spec{H_\La}$. Suppose $\eta,[\Phi]$ satisfy the hypotheses of Theorem \ref{thm:spin-gap-stability}. There exists $\ep(\gamma,T)>0$ such that for sufficiently large $\La$ and $0\leq \ep < \ep(\gamma,T)$, if $\nu , \mu \in \spec{H_\La}$ with $(\nu, \mu) \subset \mathrm{res}(H_\La)\cap [0,T]$ and $\mu - \nu > \gamma$, then the gap between $\nu$ and $\mu$ is stable. Precisely, if we denote: 
	\begin{equation*}
		\begin{split}
\gamma(\nu,\mu,\ep) = \min \set{ \la(\ep)\in \spec{H_\La(\ep)} : \la(0)\geq \mu  } & - \max \set{ \la(\ep)\in \spec{H_\La(\ep)}: \la(0) \leq \nu} 
		\end{split}
	\end{equation*}
then: 
	\begin{equation*}
		\begin{split}
\gamma(\nu,\mu,\ep) \geq (1-p\ep)\gamma - 2(q+p T+M_D)\ep>0
		\end{split}
	\end{equation*} 
for $0\leq \ep < \ep(\gamma,T)$ and $p,q$ defined:
	\begin{equation*}
		\begin{split}
p = \frac{3}{\gamma_0}  J_1 (\norm{\eta}_{F} + M_\Int)  \hspace{5mm}
q=  [C(J_3 + 4)+J_2]  ( \norm{\eta}_{F} + M_\Int)
		\end{split}
	\end{equation*}
}

\section{Stability of spectral gap in spin chains}

\subsection{Perturbations at the boundary}
Here, we make the distinction between a perturbation near the boundary and in the bulk. In this section, unless otherwise noted, we fix an interval $\La=[a,b]$ and let $\Phi$ denote the interaction $\Phi^\La$, with local Hamiltonian $\Phi_\La = \sum _{X\subset \La} \Phi(X)$.  

\medskip 

Let $D\in \nn$ define a uniform distance parameter, and denote by $\Int_D(\La)$ the relative interior $[a+D, b-D]$. The piece of the perturbation associated to $x\in \La$ is $\Phi_x = \sum _{n=1}^{R_x} \Phi(b_\La(x,n))$, and the whole perturbation is split by the relative interior: $\Phi_\La = \Phi_\La ^{D} + \Phi_{\La}^{\Int} $, where
	\begin{gather*}
\Phi_\La^D = \sum _{x\in\La \setminus \Int_D(\La)} \Phi_x \hspace{20mm} \Phi_\La^\Int = \sum _{x\in \Int_D(\La)} \Phi_x  
	\end{gather*}
are the edge and bulk perturbations, respectively. Let $\Phi^D(\La), \Phi^\Int(\La): P_f(\La)\to \ALa^s$ denote the corresponding local interactions.

\medskip 

If $x\in \Int_D(\La)$, then $n\geq r_x$ implies $\norm{ \Phi ( b_\La(x,n)} \leq \norm{\Phi }_{F} F (D)$, and so even though the bulk perturbative interaction contains terms which extend to the boundary, their contribution to the total perturbation is relatively small as a function of $D$. 

Since the Hamiltonian $H_\La+\ep \Phi_\La$ is close in operator norm to the bulk-perturbed Hamiltonian, it will suffice to prove ground state spectral gap stability for $H_\La + \ep \Phi_\La ^\Int$.  To do this, we will use a unitary decomposition method depending on \textit{spectral flow}. First proved in \cite{Michalakis:2013}, our present formulation of the following theorem using $\mc{F}$-functions comes from \cite{young:2016}. 

\subsection{Spectral flow decomposition}

Let $\Psi: P_f(I) \to \Aloc^s$ be an arbitrary interaction, $\La \subset I$, and suppose $\gamma \in (0, \gamma_0)$. Let $\ep_\La>0$ be such that $0\leq \ep  \leq \ep_\La$ implies $\gamma(H_\La(\ep))\geq \gamma$, where $H_\La(\ep) = H_\La + \ep \Psi_\La$. We may take $\ep_\La$ to be maximal. Because $\gamma(H_\La(\ep))$ is bounded below by $\gamma$ and $\ep \Psi_\La$ is uniformly bounded on $[0,\ep_\La]$, we may construct the \textit{spectral flow} (also known as \textit{quasi-adiabatic evolution}) $\alpha : [0,\ep_\La] \to \mf{A}_\La^s$, whose quasi-local properties are extensively discussed in \cite{bachmann:2012, hastings:2005}. Briefly summarizing, there exists a norm-continuous family $U(\ep)$ of unitaries such that, if $P(\ep)$ denotes the orthogonal projection onto the kernel of $H_\La(\ep)$: 
	\begin{equation}\label{eq:spectralflow}
		\begin{split}
\alpha_\ep(A) = U(\ep)^* A U(\ep) \text{ and } P(\ep) = U(\ep) P(0) U(\ep)^* .
		\end{split}
	\end{equation} 
The unitaries are the solution to $-i\frac{d}{d\ep}U(\ep)=D(\ep)U(\ep)$ with $U(0)=\1$, where the generator $D(\ep)$ is given by: 
	\begin{equation}\label{eq:generator-of-spectralflow}
		\begin{split}
D(\ep) = \int _{-\infty}^{\infty} w_\gamma(t)\int _0 ^t e^{is H_\La(\ep)} \Psi_\La e^{-i s H_\La(\ep)} ds ~dt 
		\end{split}
	\end{equation} 
for a weight function $w_\gamma \in L^1$ with compactly supported Fourier transform (see Lemma 2.3 in \cite{bachmann:2012}). Since the quasi-local properties of its generator are made clear by the expression (\ref{eq:generator-of-spectralflow}), the spectral flow automorphism transforms the perturbed Hamiltonian $H_\La(\ep)$ into a unitarily equivalent finite-volume Hamiltonian of a well-behaved, local interaction. Identifying this local interaction is the content of the unitary decomposition theorem: 

\begin{theorem}\label{thm:decomposition}
Suppose $\Psi: P_f(I) \to \Aloc^s$, satisfies a finite $\mc{F}$-norm for F and $h_\Psi(r) \geq \mc{K}r^t$ for some $\mc{K}>0$ and $t\in (0,1]$. Then for all $0\leq \ep \leq \ep_\La$: 
\begin{enumerate}
\item There exists an interaction $\Phi^1(\ep):P_f(\La) \to \mf{A}_\La^s$ such that
$\alpha_\ep( H_\La(\ep)) = H_\La + \Phi^1(\ep)$, and
\item $\Phi^1(\ep)$ is supported on the metric balls of $\La$, that is, 
\begin{equation*}
\begin{split}
\Phi^1_\La(\ep) = \sum _{x\in \La} \Phi^{1}_x(\ep)
\end{split}
\end{equation*}
where $\Phi^{1}_x(\ep) = \sum _{n=1}^{R_x} \Phi^1( b_\La(x,n),\ep) $ and each $\Phi^1( b_\La(x,n),\ep) \in \mf{A}_{b_\La(x,n)}^s$. Furthermore, for all $x\in \La$, $[P(0), \Phi_x^1(\ep)]=0$. 
\end{enumerate}
There exists a constant $C >0$, depending on the uniform bound $M$, range $R$, uniform gap $\gamma_0$ and decay parameters $\mc{K}$ and $t$, such that: 
\begin{equation*}
\begin{split}
\norm{\Phi^1(\ep)}_{F_\varphi} \leq C \ep ( \norm{\eta}_{F_\Psi} + \norm{\Psi}_{F_\Psi})
\end{split}
\end{equation*}
where $F_\varphi$ is an $\mc{F}$-function depending on $\mc{K},t,\gamma$ such that $F_\varphi( r )$ decays faster than any polynomial in $r$. 
\end{theorem} 

\begin{proof}
This reformulated statement of the original decomposition theorem found in \cite{Michalakis:2013} is proved in Theorem 6.3.4 in \cite{young:2016}, and so we record here only the precise form of $F_\vphi$. Define: 
	\begin{equation}\label{eq:mu}
		\begin{split}
\mu(r) = \bigg{ \{ } \begin{array}{l l}(e / \kappa)^\kappa & \text{ if }r \leq e^\kappa \\
r / ( \log{r})^\kappa & \text{ else } r> e^\kappa   \end{array}
		\end{split}
	\end{equation}
Define $\mc{K}_0 = \min\set{ \mc{K}, 2/7}$, and denote by $\nu_\Psi$ the Lieb-Robinson velocity for the Heisenberg dynamics generated by the interaction $\Psi$. Denote $\widetilde{\mu}(r) = \mu(\frac{\mc{K} \gamma r}{2\nu_\Psi})$ and:
	\begin{equation*}
		\begin{split}
G_\Phi(r) = e^{- \frac{\mc{K}_0}{\mc{K}} \widetilde{\mu} \circ h_\Phi(r)} F^b(r)
		\end{split}
	\end{equation*}

Then the $\mc{F}$-function in the statement of the theorem is given by: 
	\begin{equation}\label{eq:ffunction-varphi}
		\begin{split}
F_\varphi(r) = \bigg{ \{ } \begin{array}{l l} G_\Psi(0) & \text{ if } r\leq 18R +27 \\    G_\Psi(r/18 - R - 3/2)& \text{ else } r> 18R+27  \end{array} 
		\end{split}
	\end{equation} 	
\end{proof}
For the remainder of this section, let $U(\ep), \alpha_\ep$ and $\Phi^1(\ep)$ be from an application of Theorem \ref{thm:decomposition} when $\Psi$ is the bulk perturbative interaction $\Phi^\Int(\La)$ with local Hamiltonian $\Phi_\La^\Int$. 

\begin{lemma}The local operator $\Phi^1(\ep)$ can be rewritten:
\begin{equation}\label{eq:spectral-flow-decomposition}
 \Phi^1(\ep) = \Phi^2(\ep) + \Phi^3(\ep) + \om_\La(\pert{\Phi^1}(\ep)) + \mc{R}(\ep)
 \end{equation}
for terms defined:
	\begin{gather*}
 \pert{\Phi^1} (\ep)   = \sum _{x\in \Int_2(\La)} \Phi_x^1(\ep)  \\
\Phi^2(\ep)  = (\1-P)(\pert{\Phi^1}(\ep) - \om_\La(\pert{\Phi^1}(\ep)) \1)(\1-P)
\\ 
\Phi^3(\ep)  = P(\pert{\Phi^1}(\ep)-\om_\La(\pert{\Phi^1}(\ep))\1)P  \\
\mc{R}(\ep) = \Phi^1_a(\ep) + \Phi^1_{a+1}(\ep) + \Phi^1_b(\ep) + \Phi^1_{b+1}(\ep)
	\end{gather*}
\end{lemma}

\begin{proof}
This follows from a direct calculation using the fact that $[\Phi_x^1(\ep),P]=0$. 
\end{proof}

The reason for separating the boundary terms $\mc{R}(\ep)$ from the rest of the transformed perturbation is for notational convenience, since the following argument will use the fact that $\floor{r_x/2}>0$ for $x\in \Int_2(\La)$.  

\subsection{Relative form boundedness of perturbations}
The argument for relative form boundedness of the transformed perturbation $\Phi^1(\ep)$ will depend on the following two elementary lemmas.  
\begin{lemma}\label{lem:firstlemma}
Suppose $x\in \La$. For any $1\leq m \leq r_x$, 
	\begin{equation*}
		\begin{split}
		\norm{P (\Phi^1_x(\ep) - \omega_\La( \Phi^1_x(\ep)))P} \leq \norm{\Phi^1(\ep)} \bigg{(} \Omega(r_x-m) + 2 F_\vphi(m) \bigg{)}
		\end{split}
	\end{equation*}
\end{lemma}

\begin{proof}
Denote $A - \omega_\La( A) = A_0$ and $b_\La(x,n)=b_x(n)$, for brevity. For $0\leq m \leq r_x$, by linearity of $\omega_\La$: 
\begin{equation*}
P \bigg{(} \Phi_x^1(\ep)\bigg{)}_0 P = \sum _{k=1}^{R_x} P\Phi^1( b_x(k),\ep)_0 P = \sum _{k=1}^{m} P\Phi^1( b_x(k),\ep)_0 P + \sum_{k=m+1}^{R_x} P\Phi^1( b_x(k),\ep)_0 P  
\end{equation*}
We bound the two summands separately. The right summand is bounded by Proposition \ref{prop:FfunctionProperties}: 
\begin{equation*}
\sum_{k=m+1}^{R_x} \norm{P\Phi^1( b_x(k),\ep)_0 P} \leq 2 \norm{\Phi^1(\ep)}_{F_\vphi} F_\vphi(m)   
\end{equation*}
The left summand is bounded by local topological quantum order and the $\mc{F}$-norm: 
\begin{equation*}
\begin{split}
 \sum _{k=1}^{m} \norm{ P \Phi^1(b_x(k),\ep)_0P} & \leq \sum _{k=1}^{m} \Omega(r_x-k) \norm{\Phi^1(b_x(k),\ep)} \\
& \leq  \Omega(r_x-m)\norm{ \Phi^1(\ep)}_{F_\vphi} 
\end{split}
\end{equation*}
Combining these bounds proves the lemma. 
\end{proof}

The next lemma uses the cut-off function $z_x$ defined in (\ref{eq:distances}).

\begin{lemma}\label{lem:secondlemma}
Suppose $x\in \Int_2(\La)$. If  $1\leq m \leq r_x$ and $m\leq n \leq R_x$, then:
	\begin{equation*}
		\begin{split}
\bigg{\Vert} \sum_{k=1}^m \bigg{(} \Phi^1(b_x(k),\ep) \bigg{)}_0 P_{b_x(n)} \bigg{\Vert}\leq \norm{\Phi^1(\ep)}_{F_\vphi}\bigg{[} 5 \Omega( z_x(n)-m)^{1/2}  + 4F_\vphi(m) \bigg{]} 
		\end{split}
	\end{equation*}
\end{lemma}

\begin{proof}
Suppose $A \in \A^s_{b_x(k)}$. The $C^*$-identity and LTQO imply: 
	\begin{equation*}
		\begin{split}
 \bigg{|} \norm{A P_{b_x(n)}} - \norm{AP} \bigg{|}^2  \leq \bigg{|} \norm{AP_{b_x(n)}}^2 - \norm{AP}^2 \bigg{|} \leq 2 \norm{A}^2 \Omega(z_x(n)-m) 
		\end{split}
	\end{equation*}
In the case $A = \sum _{k=1} ^m  \Phi^1(b_x(k),\ep)_0$, the above bound and Proposition \ref{prop:FfunctionProperties} imply:
	\begin{equation*}
		\begin{split}
\norm{A P_{b_x(n)}} \leq 4 \norm{\Phi^1(\ep)}_{F_\vphi} \Omega(z_x(n)-m)^{1/2}+\norm{AP}
		\end{split}
	\end{equation*} 
By Theorem \ref{thm:decomposition}, $\Phi^1_x(\ep)$ commutes with $P$. So, using Lemma \ref{lem:firstlemma}, we get: 
	\begin{equation*}
		\begin{split}
\norm{AP} & \leq \norm{P \Phi^1_x(\ep)_0 P} + 2\sum_{k=m+1}^{R_x} \norm{\Phi^1(b_x(k),\ep)} \\
& \leq \norm{\Phi^1(\ep)}_{F_\vphi} \bigg{[}\Omega(r_x-m) + 4F_\vphi(m) \bigg{]}
		\end{split}
	\end{equation*}
\end{proof}

Proposition \ref{prop:final-decomposition} uses a finite resolution of identity $\set{E_n^x}$ defined at each site $x\in \Int_2(\La)$ by: 
	\begin{equation*}
		\begin{split}
E^x_{n} = \Bigg{ \{ }\begin{array}{l l} \1 - P_{b_x(1)} & \text{ if }n=1 \\ 
P_{b_x(n-1)} - P_{b_x(n)} & \text{ if } 1<n\leq r_x \\ P_{b_x(r_x)} - P & \text{ if }n=r_x+1\\
P & \text{ else }n=r_x+2  \end{array}
		\end{split}
	\end{equation*}

\begin{lemma}\label{lem:resolution-of-identity}
The family $\set{ E_n^x}$ has the properties: 
	\begin{equation*}
		\begin{split}
& 1. ~ \sum _{k=1}^{r_x+2} E_k^x = \1  \text{ and } \sum _{k=1}^{m} E_k^x = \bigg{ \{ } \begin{array}{l l} \1 - P_{b_x(m)} & \text{ if } 1 \leq m \leq r_x \\ \1 - P & \text{ if } m = r_x+1\end{array} \\
& 2. ~ P_{b_x(k)} E_k^x = 0 \text{ for } k \leq r_x  
		\end{split}
	\end{equation*}
\end{lemma}

\begin{proof}
We only comment that the second property follows from the frustration free assumption on $\eta$. 
\end{proof}

\begin{proposition}\label{prop:final-decomposition}
Let $x\in  \Int_2(\La)$ and $0\leq \ep  \leq \ep_\La$. There exist local operators $\Theta_\beta^x(n,\ep)$, for $3\leq n \leq r_x$, and operator $\Theta_\alpha^x(\ep)$ such that:
	\begin{equation*}
		\begin{split}
\Phi_x^1(\ep)_0 = \sum _{n=3}^{r_x} \Theta_\beta^x(n,\ep) + \Theta_\alpha^x(\ep)
		\end{split}
	\end{equation*}
Furthermore, $P_{b_x(n)}\Theta_\beta^x(n,\ep)=0$, and $\Theta^x_\beta(n,\ep)$ and $\Theta^x_\alpha(\ep)$ decay rapidly: 
	\begin{equation*}
		\begin{split}
\norm{\Theta_\beta^x(n,\ep)} & \leq 20 \norm{\Phi^1(\ep)}_{F_\vphi}   \bigg{[} \Omega\bigg{(}\frac{n-1}{2}\bigg{)}^{\frac{1}{2}} +  F_\vphi \bigg{(} \frac{n-3}{2} \bigg{)} \bigg{]} \\
\norm{\Theta_\alpha^x(\ep)} & \leq 20 \norm{\Phi^1(\ep)}_{F_\vphi}\bigg{[}  \Omega\bigg{(}\frac{r_x-1}{2}\bigg{)}^{\frac{1}{2}} +  F_\vphi \bigg{(} \frac{r_x-3}{2}\bigg{)} \bigg{]}
		\end{split}
	\end{equation*}
\end{proposition}

\begin{proof}
Fix $x\in (a,b)$ and $\ep \in [0,\ep_\La]$. Abbreviate $Q=\1-P$ and $\Phi_k^1 = \Phi^1(b_x(k),\ep)_0$, i.e.
	\begin{equation*}
		\begin{split}
\Phi^1_x(\ep)_0 = \sum _{k=1}^{R_x} Q\Phi^1_k Q
		\end{split}
	\end{equation*}
Define a ``cut-off" parameter $n_x = \floor{\frac{r_x}{2}}$ and split $\Phi^1_x(\ep)_0$ into two sums: 
	\begin{equation}\label{eq:prelim-decomp}
		\begin{split}
\Phi^1_x(\ep)_0 = \sum _{k=1}^{n_x} Q \Phi^1_k Q + \sum _{k=n_x+1}^{R_x} Q \Phi_x^1 Q 
		\end{split}
	\end{equation}
The tail $\mu_\alpha^x = \sum _{k=n_x+1}^{R_x} Q\Phi^1_k Q$ can be bounded above in operator norm by using LTQO, so we turn our attention to the other summand. Denote by $Q_{b_x(l)}$ the complement projection $\1 - P_{b_x(l)}$. Using the resolution $\set{E_n}$ at $x$, we rewrite $Q\Phi_k^1Q$ for all $1\leq k \leq n_x$ as:
	\begin{equation}\label{eq:sum-rewrite}
		\begin{split}
 Q\Phi^1_kQ = Q_{b_x(2k)} \Phi^1_k Q_{b_x(2k)} + \sum _{n=2k+1}^{r_x+1} \bigg{[} E_n \Phi_k^1 \bigg{(} \sum _{m=1}^{n-1} E_m \bigg{)} + \bigg{(} \sum _{m=1}^{n} E_m \bigg{)} \Phi^1_k E_n  \bigg{]}
		\end{split}
	\end{equation}
Define the following terms to organize the summands in (\ref{eq:sum-rewrite}):
\begin{gather*}
\nu_\alpha^x(k)  = E_{r_x+1} \Phi_k^1 Q_{b_x(r_x)} + Q\Phi^1_k E_{r_x+1} \hspace{10mm} 
\theta_\beta^x(n,k)  = E_n \Phi ^1_k Q_{b_x(n-1)} + Q_{b_x(n)}\Phi^1_kE_n \\ 
\tau_\beta^x(2k) = Q_{b_x(2k)} \Phi_k^1 Q_{b_x(2k)} 
\end{gather*}
so that:
	\begin{equation*}\label{eq:interchange-orders}
		\begin{split}
Q\Phi_k^1Q & = \nu_\alpha^x(k) + \tau_\beta^x(2k) + \sum _{n=2k+1}^{r_x} \theta_\beta^x(n,k).
		\end{split}
	\end{equation*}
For convenience, extend $\tau_\beta^x(m)$ to previously undefined $m$ by declaring $\tau_\beta^x(m)=0$. The derivation of the $\Theta_\beta^x(\ep, n), \Theta^x_\alpha(\ep)$ operators will result from an interchange of order for the summation of terms in (\ref{eq:interchange-orders}) over $n$ and $k$. The following definition for $\Theta_\beta^x(n,\ep)$ accounts for the parity of $r_x$: 
	\begin{equation*}
		\begin{split}
\forall 3\leq n < r_x: ~ \Theta_\beta^x(n,\ep) & = \bigg{[} \sum_{k=1}^{\floor{\frac{n-1}{2}}} \theta_\beta^x(n,k) \bigg{]} + \tau_\beta^x(n) \\
\Theta_\beta^x(r_x,\ep) & =  \sum _{k=1}^{\floor{\frac{r_x}{2}}} \theta_\beta^x(r_x,k)+\tau_\beta^x(r_x) 
		\end{split}
	\end{equation*}
Then: 
	\begin{equation*}
		\begin{split}
Q\Phi^1_x(\ep)_0 Q = \sum _{k=1}^{R_x} Q\Phi^1_k Q = \sum _{n=3}^{r_x} \Theta_\beta^x(n,\ep) + \Theta_\alpha^x(\ep) 
		\end{split}
	\end{equation*}
where $\Theta_\alpha^x(\ep) = \mu_\alpha^x + \sum _{k=1}^{n_x} \nu_\alpha^x(k)$. Next, the frustration free property of $H_\La$ implies that $\ker(H_{b_x(n)}) \subset \ker(H_{b_x(n-1)})$, and so:
	\begin{equation}\label{eq:strongcommutativity}
\forall 3\leq n \leq r_x:  P_{b_x(n)} \Theta_\beta^x(n,\ep) = \Theta_\beta^x(n,\ep) P_{b_x(n)} = 0
 	\end{equation}
Furthermore, we have the following bounds on operator norm, for all $x\in \Int_2(\La)$ and $ 3 \leq n < r_x$, by Lemma \ref{lem:secondlemma} and Proposition \ref{prop:FfunctionProperties}:
	\begin{equation}\label{eq:formrelative-bounds-opnorm}
		\begin{split}
 \norm{\Theta_\beta^x(n,\ep)} & \leq \norm{ (\sum _{k=1}^{\floor{\frac{n-1}{2}}} \Phi^1_k )^* E_n } + \norm{\sum _{k=1}^{\floor{\frac{n-1}{2}}} \Phi_k^1 E_n} + \norm{\tau_\beta^x(n)}  \\ & \leq 20 \norm{\Phi^1(\ep)}_{F_\vphi} \bigg{[} \Omega\bigg{(} \frac{n-1}{2}\bigg{)}^{\frac{1}{2}} + F_\vphi \bigg{(} \frac{n-3}{2}\bigg{)} \bigg{]}  \\
\max\set{\norm{\Theta_\beta^x(r_x,\ep)},\norm{\Theta^x_\alpha(\ep)}} & \leq 20 \norm{\Phi^1(\ep)}_{F_\vphi} \bigg{[}\Omega\bigg{(} \frac{r_x-1}{2}\bigg{)}^{\frac{1}{2}} + F_\vphi\bigg{(} \frac{r_x-3}{2}\bigg{)} \bigg{]}
	\end{split}
	\end{equation} 
\end{proof} 
%
%
%
%
%
%
%
%
%
%
Now, we define several quantities which will appear in the derivation of the form bound. Note the weight function $e^{-h_\vphi(x)}$ of $F_\vphi$ is bounded above by $1$ on its domain. So any expression in $F_\vphi$ is bounded above by the corresponding sum using the shifted base $\mc{F}$-function 
	\begin{equation}\label{eq:shift-ffunction}
		\begin{split}
F_0(r)=F^b(r/18 - R - 3/2)
		\end{split}
	\end{equation}
from (\ref{eq:base-f}) and (\ref{eq:ffunction-varphi}). Define:
	\begin{equation}\label{eq:independentbound}
		\begin{split}
\kappa(n,\ep) = 20 C \ep( \norm{\eta}_{F} + \norm{\Phi^\Int(\La)}_{F}) \bigg{[} \Omega\bigg{(} \frac{n-1}{2} \bigg{)}^{\frac{1}{2}} + F_0\bigg{(}\frac{n-3}{2}\bigg{)} \bigg{]}
		\end{split}
	\end{equation} 
$\kappa(n,\ep)$ does not depend on either $\La$ or the lower bound $\gamma$ on the instantaneous gap, and the inequalities from (\ref{eq:formrelative-bounds-opnorm}) are rewritten: 
	\begin{equation*}
		\begin{split}
\norm{\Theta_\beta^x(n,\ep)} \leq \kappa(n,\ep) \hspace{20mm} \norm{\Theta_\alpha^x(\ep)} \leq \kappa( r_x, \ep)
		\end{split}
	\end{equation*}
Lastly, we see by the assumed decay of $\Omega$ that the following sums are finite: 
	\begin{equation}\label{eq:Jconstants}
		\begin{split}
J_1 & = \sum _{n \in \zz} 20C |n|  [ \Omega( (|n|-1)/2)^{1/2} + F_0((|n|-3)/2)] \\
J_2 & = \sum _{n \in \zz} 20C [ \Omega( (|n|-1)/2)^{1/2} + F_0((|n|-3)/2)]  
		\end{split}
	\end{equation}
The following argument for concluding form boundedness is essentially due to \cite{Michalakis:2013}, modified to work with the boundary terms introduced by Proposition \ref{prop:final-decomposition}. We divide a large part of the Hamiltonian with respect to a convenient partition of $ \Int_2(\La)$. For $n\in \nn$, define the relation $x\sim_n y$ if and only if $x-y\in (2n+1)\zz$. Index each of the parts $\La_n^i$ of $ \Int_2(\La) / \sim_n $ by a representative $i\in I(n) \subset \Int_2(\La)$. Note that the cardinality of $I(n)$ is roughly bounded above by $3n$. The corresponding parts of the Hamiltonian are defined: 
	\begin{equation*}
		\begin{split}
H_n^i  = \sum _{x\in \La_n^i} H_{b_x(n)} \hspace{20mm}
\Phi_n^i  = \sum _{x\in \La_n^i}\Theta_\beta^x(n,\ep)
		\end{split}
	\end{equation*}
By definition of the $\Theta_\beta^x(n,\ep)$ operators, $\Phi^2(\ep) = \sum _{n,i} \Phi^i_n$. In order to compare $H_n^i$ to $\Phi_n^i$, we use a resolution of identity from \cite{Michalakis:2013}, whose properties we record here: 

\begin{lemma}\label{lem:MZ-res}
For a configuration $\sig: \La_n^i \to \set{0,1}$, define the projection $S_n^i(\sig) = \prod _{x\in \La_n^i} \sig_x Q_{b_x(n)} + (1-\sig_x)P_{b_x(n)}$. Then: 
	\begin{equation*}
		\begin{split}
& 1. ~ \sum _{\sig: \La_n^i \to \set{0,1} } S_n^i(\sig)=\1 \\
& 2. ~ S_n^i(\sig)S_n^i(\sig') = \delta_{\sig,\sig'} S_n^i(\sig)\\
& 3. ~ \text{For all $x\in \La_n^i$}, ~ [\Theta_\beta^x(n,\ep), S_n^i(\sig)]=0
		\end{split}
	\end{equation*}
\end{lemma}

\begin{proof}
These properties follow immediately from the fact that $P_{b_x(n)}\Theta^x_\beta(n,\ep)=0$ and that $x\sim_n y $ implies $b_x(n) \cap b_y(n)=\emptyset$. 
\end{proof} 

\begin{proposition}\label{prop:preliminary-relbound}
Suppose $\diam{\La}> \max\set{4,R}$. There exist constants $\delta,\beta>0$, dependent on $\norm{\Phi^\Int(\La)}_{F}$, such that $0 \leq \ep \leq \ep_\La$ implies, for all $v \in \mf{H}_\La$: 
	\begin{equation}\label{eq:preliminary-relbound}
		\begin{split}
|\ip{v, \Phi^2(\ep) v}| \leq \delta\ep \norm{v}^2 + \beta\ep \ip{v,H_\La v}
		\end{split}
	\end{equation}
Precisely, we may choose: 
$$\delta=J_2(\norm{\eta}_{F} + \norm{\Phi^\Int(\La)}_{F})  \text{ and }\beta=\frac{3}{\gamma_0}  J_1 (\norm{\eta}_{F} + \norm{\Phi^\Int(\La)}_{F})  .$$  
\end{proposition}

\begin{proof}
Denote $d_\La = \diam{\La}$. For any $x\in \Int_2(\La)$, if $n> r_x$, say that $\Theta_\beta^x(n,\ep)=0$. Suppose $u\in \mf{H}_\La$. Then by Proposition \ref{prop:final-decomposition}: 
	\begin{equation*}
		\begin{split}
|\ip{u,\Phi^2(\ep) u }| \leq | \ip{u, \sum _{n=3}^{d_\La} \sum _{i\in I(n)} \Phi^i_n u } | + \sum _{x\in \Int_2(\La)} \kappa(r_x,\ep) \norm{u}^2
		\end{split}
	\end{equation*}
The second term $\sum _{x\in \Int_2(\La)}\kappa(r_x,\ep)$ is bounded above by the constants in (\ref{eq:Jconstants}), so we focus on the first summand. Since $[\Phi^i_n, S_n^i(\sig)]=0$: 
	\begin{equation}\label{eq:norm-calculations}
		\begin{split}
| \ip{u, \sum _{n=3}^{d_\La} \sum _{i\in I(n)} \Phi^i_n u } | & \leq \sum _{n=3}^{d_\La} |\ip{u,\sum _{i\in I(n)} \Phi^i_n \bigg{[} \sum _\sig S_n^i (\sig) \bigg{]} u } | \\
& \leq \sum _{n=3}^{d_\La} \sum_{i\in I(n)} \sum_{\sig:\La_n^i\to \set{0,1}} \sum _{\substack{x\in \La_n^i }} \norm{S_n^i(\sig)\Theta_\beta^x(n,\ep)} \ip{u,S_n^i(\sig)u} \\
& \leq \sum_{n=3}^{d_\La}  \sum_{i\in I(n)} \frac{\kappa(n,\ep)}{\gamma_0}  \sum _{\substack{x\in \La_n^i}} ~~ \sum _{\substack{\sig:\La^i_n\to \set{0,1} \\ \sig_x=1}}  \gamma_0 \ip{ u, S_n^i(\sig)u} \\
& = \sum _{n=3}^{d_\La} \sum _{i\in I(n)} \frac{\kappa(n,\ep)}{\gamma_0} \sum _{x\in \La^i_n} \ip{ u, \gamma_0 Q_{b_x(n)} u} \\
& \leq \sum _{n=3}^{d_\La} \frac{3n \kappa(n,\ep)}{\gamma_0} \ip{u, H_\La u} 
		\end{split}
	\end{equation}	
Hence: 
	\begin{equation*}
		\begin{split}
|\ip{u, \Phi^2(\ep)u}| & \leq \bigg{[}\sum _{x\in \Int_2(\La)} \kappa(r_x,\ep) \bigg{]} \norm{u}^2 + \bigg{[} \sum _{n=3}^{d_\La} \frac{3n\kappa(n,\ep)}{\gamma_0} \bigg{]}\ip{u, H_\La u}  \\
&  \leq J_2(\norm{\eta}_{F} + \norm{\Phi^\Int(\La)}_{F})  \norm{u}^2 + \ \frac{3}{\gamma_0} J_1 (\norm{\eta}_{F} + \norm{\Phi^\Int(\La)}_{F}) \ep \ip{u, H_\La u}
		\end{split}
	\end{equation*}
	
	\end{proof}

\begin{corollary}\label{cor:norm-boundedness}
There exists a constant $\alpha$, dependent on $\norm{\Phi^\Int(\La)}_{F}$, such that $0 \leq \ep \leq \ep_\La$ and $\diam{\La}>\max\set{4,R}$ imply: 
	\begin{equation*}
		\begin{split}
		\forall u\in \mf{H}_\La: ~|\ip{ u,(\Phi^2(\ep)+ \Phi^3(\ep) + \mc{R}(\ep) ) u }| \leq \alpha \ep \norm{u}^2 + \beta \ep \ip{u,H_\La u}		
		\end{split}
	\end{equation*}
Precisely, we may take $\alpha =  C ( \norm{\eta}_{F} + \norm{\Phi^{\Int}(\La)}_{F}) [J_3 + 4] + \delta$.
\end{corollary}

\begin{proof}
Suppose $x\in \Int_2(\La)$. Set $m =  \lfloor \frac{r_x}{2} \rfloor$ in an application of Lemma \ref{lem:firstlemma} to show: 
	\begin{equation*}
		\begin{split}
\norm{P(\Phi_x^1(\ep))_0P} \leq \norm{\Phi^1(\ep)}_{F_\vphi} [ \Omega( r_x/2) + 2F_\vphi( \floor{r_x/2})]
		\end{split}
	\end{equation*}
But by the decay of $\Omega$ and $F_0$, we have that the following sum is finite: 
	\begin{equation}\label{eq:third-J-constant}
		\begin{split}
J_3 = \sum _{z\in \zz} \Omega( |z|/2) + 2 F_0(\floor{|z|/2})
		\end{split}
	\end{equation}
And, summing over $x\in \Int_2(\La)$:
	\begin{equation*}
		\begin{split}
\norm{\Phi^3(\ep)} \leq \sum _{x\in\Int_2(\La)} \norm{P(\Phi_x^1(\ep))_0P} \leq \norm{\Phi^1(\ep)}_{F_\vphi} J_3
		\end{split}
	\end{equation*}
Next, it is straightforward to apply Proposition \ref{prop:FfunctionProperties} to $\mc{R}(\ep)$ to get an upper bound on the norm:
	\begin{equation*}
		\begin{split}
\norm{\mc{R}(\ep)}  \leq 4 \norm{\Phi^1(\ep)}_{F_\vphi}
		\end{split}
	\end{equation*}  
\end{proof}

Until now, all estimates have been expressed using a local bound $\norm{\Phi^\Int(\La)}_{F}$ on the strength of the bulk perturbation for fixed $\La$. In order to obtain volume independent lower bounds on the spectral gap, we use the following uniform quantity:
	\begin{equation*}
		\begin{split}
M_\Int = \sup_\La \set{ \norm{\Phi^\Int(\La)}_{F} : \diam{\La} > \max \set{ 2D, R}}
		\end{split}
	\end{equation*}
	
\begin{proposition}\label{prop:interior-stability}
There exist $\ep_\Int>0$ and constant $m>0$ such that $0\leq \ep < \ep_\Int$ and $\diam{\La}>\set{4,R}$ imply:
	\begin{equation*}
		\begin{split}
\gamma(H_\La + \ep \Phi^\Int_\La) \geq \gamma_0 - m\ep >0
		\end{split}
	\end{equation*}
The constants $\ep_\Int$ and $m$ can be taken as the following expressions:
	\begin{gather*}
m   = \bigg{(} 3J_1 + 2J_2 + C(J_3 + 8) \bigg{)} ( \norm{\eta}_{F} + M_\Int)\\
\ep_\Int  = \min \set{ 1, \frac{\gamma_0}{m}}
	\end{gather*}
\end{proposition}
	
\begin{proof}
Let $\gamma \in (0,\gamma_0)$. For fixed $\La$ with $\diam{\La}>\max\set{4,R}$, there exists $\ep_\La>0$ such that for all $0\leq \ep \leq \ep_\La$, $\gamma(H_\La+ \ep \Phi^\Int_\La) \geq \gamma$. By continuity of the eigenvalue functions, we may assume $\ep_\La$ is maximal, i.e. either $\ep_\La = 1$ or there exists $c>0$ such that for all $\mu \in (\ep_\La, \ep_\La +c)$, $\gamma(H_\La+ \mu \Phi^\Int_\La) < \gamma$. 

\medskip 

Since the gap does not close on $[0,\ep_\La]$, we use the spectral flow decomposition (\ref{eq:spectral-flow-decomposition}) to transform $H_\La+ \ep \Phi^\Int_\La$ by unitaries and a shift in the spectrum: 
	\begin{equation*}
		\begin{split}
\alpha_\ep (H_\La+ \ep \Phi^\Int_\La) - \omega_\La ( \widetilde{ \Phi^1(\ep)}) = H_\La + \Phi^2(\ep) + \Phi^3(\ep) + \mc{R}(\ep)
		\end{split}
	\end{equation*}
But by Corollary \ref{cor:norm-boundedness}, if $\ep \leq \ep_\La$, then $\Phi(\ep)=\Phi^2(\ep) + \Phi^3(\ep) + \mc{R}(\ep)$ is $H_\La$-bounded. Now, by the relation $P(\ep) = U(\ep) P(0)U(\ep)^*$ in (\ref{eq:spectralflow}), the span of the eigenvectors to the $0$-group of $H_\La + \Phi(\ep)$ is exactly $\ker(H_\La)$. So, if $\la$ is in the $0$-group, which we will denote by $\specc(0,\ep)$, then there exists a unit norm $u\in \ker(H_\La)$ such that:
	\begin{equation}\label{eq:lower-diameter}
		\begin{split}
|\la| & = |\ip{u,(H_\La+\Phi(\ep) )u}|  \leq \alpha \ep
		\end{split}
	\end{equation}
Next, define $\ep_1>0$ as the solution to $h(\ep)=\gamma$, where $h$ is defined:
	\begin{equation*}
		\begin{split}
h(\ep)=(1-\beta\ep)\gamma_0 - \delta\ep - 4C\ep ( \norm{\eta}_{F} + M_\Int) - \alpha\ep 
		\end{split}
	\end{equation*}
Set $\ep_\gamma = \min \set{ \ep_1, 1}$. Combining (\ref{eq:lower-diameter}) and (\ref{cor:norm-boundedness}), we see that if $0\leq \ep < \min \set{\ep_\gamma, \ep_\La}$, then: 
	\begin{equation*}
		\begin{split}
\gamma(H_\La(\ep)) & = \min _{v\in \ker(H_\La)^\perp: \norm{v}=1} \ip{ v, [H_\La + \Phi^2(\ep) + \mc{R}(\ep)] v} - \max \specc(0,\ep) \\
& \geq h(\ep) \\
& > \gamma		\end{split}
	\end{equation*}
By maximality, either $\ep_\La =1$ or $\gamma(H_\La+ \ep_\La \Phi^\Int_\La) = \gamma$. Hence $\ep_\gamma \leq \ep_\La$ necessarily and $\gamma(H_\La+ \ep \Phi^\Int_\La) \geq h(\ep) >\gamma$ for all $\ep < \ep_\gamma$. But now, $\gamma$ was arbitrarily smaller than $\gamma_0$. Set: 
$$\ep_{\Int}= \sup \set{\ep_\gamma: \gamma \in (0,\gamma_0)}. $$
Evidently $\ep_{\Int}$ does not depend on $\La$, and if $0 \leq \ep < \ep_{\Int}$, then: 
	\begin{equation*}
		\begin{split}
\gamma(H_\La + \ep \Phi^\Int _\La) \geq h(\ep) = \gamma_0 - m\ep  > 0 
		\end{split}
	\end{equation*}
where the constant: 
	\begin{equation*}
		\begin{split}
m =  \big{(}3J_1 + 2J_2 + C( J_3 + 8)\big{)}( \norm{\eta}_{F} + M_\Int)
		\end{split}
	\end{equation*} 
comes from rewriting the lower bound $h(\ep)$ as a linear equation of $\ep$.  
\end{proof} 
\noindent Denote by $M_D$ the following finite uniform bound on the strength of the edge perturbations:
	\begin{equation*}
		\begin{split}
M_D = \sup_\La \set{ \norm{ \Phi^D_\La }: \diam{\La}> \max\set{ 2D, R}}
		\end{split}
	\end{equation*}
We remark that $M_\Int$ and $m$ are defined in terms of $\mc{F}$-function decay, while $M_D$ is defined in terms of the operator norm. 
\begin{theorem}[Ground state gap stability for spin chains]\label{thm:spin-gap-stability}
Suppose $\eta: P_f(\zz) \to \Aloc^s$ has LTQO with $\Omega(n) \leq n^{-\nu}$, for $\nu>4$, and there exist $K>0$, $s\in (0,1]$ such that $h_\Phi$ satisfies $h_\Phi(r) \geq Kr^s$. Then there exists $\ep(\gamma_0)>0$ such that $0\leq \ep < \ep(\gamma_0)$ and $\diam{\La}>\max\set{2D,R}$ imply: 
	\begin{equation*}
		\begin{split}
\gamma(H_\La(\ep)) \geq \gamma_0 - (m+2M_D) \ep >0
		\end{split}
	\end{equation*}
The constant $ \ep(\gamma_0)$ can be taken as:
\begin{equation}\label{eq:specific-constants}
	\begin{split}
\ep(\gamma_0) & = \min \set{1, \frac{\gamma_0}{m + 2 M_D}}
	\end{split}
	\end{equation}
\end{theorem}

\begin{proof}
Considering $\ep \Phi^D_\La$ as a perturbation of $H+\ep \Phi^\Int_\La$, the spectrum of $H+\ep \Phi^D_\La + \ep \Phi^\Int _\La$ must be contained in the compact neighborhood: 
$$\mc{O}_\La(\ep) = \set{ r\in \rr: d(r, \spec{H + \ep \Phi^\Int_\La}) \leq \norm{ \ep \Phi^D_\La} }$$
That is, 
$$\gamma(H_\La(\ep)) \geq \gamma(H_\La + \ep \Phi^\Int_\La) - 2 \norm{ \ep \Phi^D_\La} \geq \gamma_0 - (m + 2M_D)\ep $$
\end{proof}

Since the stability theorem guarantees a $\La$-independent neighborhood of $0$ where a relative form bound of the perturbation will hold, we can also conclude stability of spectral gaps which are located higher in the spectrum. 

\begin{proposition}\label{prop:higher-gaps}
Let $T,\gamma>0$, and denote $\mathrm{res}(H_\La) = \cc \setminus \spec{H_\La}$. Suppose $\eta,[\Phi]$ satisfy the hypotheses of Theorem \ref{thm:spin-gap-stability}. There exists $\ep(\gamma,T)>0$ such that for sufficiently large $\La$ and $0\leq \ep < \ep(\gamma,T)$, if $\nu , \mu \in \spec{H_\La}$ with $(\nu, \mu) \subset \mathrm{res}(H_\La)\cap [0,T]$ and $\mu - \nu > \gamma$, then the gap between $\nu$ and $\mu$ is stable. Precisely, if we denote: 
	\begin{equation*}
		\begin{split}
\gamma(\nu,\mu,\ep) = \min \set{ \la(\ep)\in \spec{H_\La(\ep)} : \la(0)\geq \mu  } & - \max \set{ \la(\ep)\in \spec{H_\La(\ep)}: \la(0) \leq \nu} 
		\end{split}
	\end{equation*}
then: 
	\begin{equation*}
		\begin{split}
\gamma(\nu,\mu,\ep) \geq (1-p\ep)\gamma - 2(q+p T+M_D)\ep>0
		\end{split}
	\end{equation*} 
for $p,q$ defined:
	\begin{equation}
		\begin{split}
p = \frac{3}{\gamma_0}  J_1 (\norm{\eta}_{F} + M_\Int)  \hspace{5mm}
q=  ( \norm{\eta}_{F} + M_\Int) [C(J_3 + 4)+J_2] 
		\end{split}
	\end{equation}
\end{proposition}

\begin{proof}
Let $\Phi(\ep)$ be defined as in Proposition \ref{prop:interior-stability}, for $0 \leq \ep < \ep_\Int.$ By Proposition \ref{cor:norm-boundedness}, for all $u\in \mf{H}_\La$:
	\begin{equation*}
		\begin{split}
|\ip{u, \Phi(\ep) u} | \leq p\ep \ip{ u, H _\La u} + q\ep \norm{u}^2 
		\end{split}
	\end{equation*}
Let $z= \frac{\nu + \mu}{2}$ and denote $R_\zeta (\ep') = (\zeta - H_\La - \Phi(\ep'))^{-1}$, with $R_\zeta = R_\zeta(0)$. Let $U$ denote the polar unitary such that $R_z = U |R_z|$.  Since $R_z$ is self-adjoint, $|R_z|U^* = U|R_z|$, and so for unit norm $u$: 
	\begin{equation}\label{eq:neumann-radius-of-convergence}
		\begin{split}
\sup_{\norm{w}=1}| \ip{w,|R_z|^{1/2} U^* \Phi(\ep) |R_z|^{1/2} u} |  & \leq \norm{ |R_z|^{1/2} \Phi(\ep) |R_z|^{1/2} } \\ & \leq \sup_{\norm{v}=1} \bigg{[} q\ep \norm{ |R_z|^{1/2} v}^2 + p\ep \ip{ v, H_\La |R_z| v}   \bigg{]}   
	\end{split}
	\end{equation}
That is, for sufficiently small $\ep$:
\begin{equation*}
		\begin{split}
\norm{ |R_z|^{1/2} U^* \Phi(\ep)  |R_z|^{1/2} } \leq q\ep \norm{R_z}+ p \ep (1+|z| \norm{R_z}) < 1
		\end{split}
	\end{equation*}
and by the expansion: 
	\begin{equation*}
		\begin{split}
R_z(\ep) = U |z-H|^{1/2} ( \1 - |R_z|^{1/2} U^* \Phi(\ep) |R_z|^{1/2} ) |z-H|^{1/2}
		\end{split}
	\end{equation*}
we derive the lower bound: 
	\begin{equation*}
		\begin{split}
 d(z, \spec{H_\La + \Phi(\ep)}) \geq (1-p \ep) \gamma - 2 (q+p T) \ep .
		\end{split}
	\end{equation*}
Hence for sufficiently small $\ep$, independently of sufficiently large $\La$,
	\begin{equation*}
		\begin{split}
\gamma(\nu,\mu, \ep) \geq (1-p \ep) \gamma - 2(q+pT + M_D)\ep > 0
		\end{split}
	\end{equation*} 
\end{proof}

\subsection{The thermodynamic limit}\label{sec:TL}

So far, we have studied finite spin chains and shown that, under a set of general assumptions, the group of eigenvalues continuously connected 
to the ground state energy of a finite frustration-free Hamiltonian remains separated by a gap from the rest of the spectrum, uniformly in the length 
of the chain and as long as the perturbations are not too large. We now want to show that the states associated with this group of eigenvalues 
all converge to a ground state of the model in the thermodynamic limit. The lower bound for the gap of finite chains is then also a lower bound for
the gap above those ground states of the infinite chain.

For concreteness, we consider Hamiltonians of the form \eq{eq:deform}, where $\eta$ satisfies the assumption set out in Section \ref{sec:assumptions}, 
and $[\Phi]=\{ \Phi^\Lambda\mid \Lambda\in\cP_f(\Ir)\}$ is a family of perturbations given in terms of interactions $\Phi,\Phi^b\in \cB_F$ and a few parameters 
that define the boundary conditions. Specifically, consider intervals $\Lambda\subset\Ir$ of the form $[-a,b]$, $a,b\geq 0$, and for any $D\geq 0$, let 
$\Int_D(\La)=[-a+D,b-D]$. Let $\partial$ denote the triple of parameters $(D_1,D_2,s)$, $D_1,D_2\geq 0, s\in [0,1]$ and consider
\be\label{curveham}
H^\partial_\Lambda(\epsilon) = \sum_{X\subset\Lambda}\eta(X) + \epsilon \left( \sum_{X\subset\Lambda_{D_1}}\Phi(X) + s \sum_{X\subset(\Lambda\setminus\Lambda_{D_2})}\Phi^b(X)\right).
\ee
This form of the Hamiltonian covers a broad range of perturbations and boundary conditions. The dynamics generated by $H^\partial_\Lambda(\epsilon)$ is the one-parameter 
group of automorphism $\tau_t^{H^\partial_\Lambda(\epsilon)}$.

As explained in Section \ref{sec:appendixTL}, if we take, for example, $\Lambda_n=[-a_n,b_n] $, $s_n\in[0,1]$ arbitrary, and $D_{1,n}, D_{2,n}$ such that $\min(a_n,b_n)-\max(D_{1,n}, D_{2,n})\to\infty$, 
then there is a strongly continuous group of automorphisms $\tau^{\epsilon}_t, t\in\Rl$ on $\cA_\Ir$ such that
\begin{equation}\label{cont_bd_iv} 
\lim_{n\to\infty} \| \tau_t^{H^{\partial_n}_{\Lambda_n}(\epsilon)}(A) - \tau^\epsilon_t(A) \| =0, \mbox{ for all } A\in \cA^{\rm loc}_\Ir.
\end{equation} 

If we take $\epsilon\in [0,\epsilon(\gamma_0))$, with $\epsilon(\gamma_0)$ as in Theorem \ref{thm:spin-gap-stability}, we have a uniform gap
separating the lower portion of the spectrum of  $H^{\partial_n}_{\Lambda_n}(\epsilon)$, denoted by $\specc_{0,\La_n}(\ep)$ in \eq{def_sp0}, 
and the rest of the spectrum. The following results provides an estimate of $\diam{\specc_{0,\La_n}(\ep)}$. For simplicity, let $\La_n = [-n,n]$ for the remainder of the section. 

\begin{lemma}\label{lem:sublinear-locality}
In the assumptions of above, choose $s_n=0$ and put $D_{1,n} = D_n$. Then, there exists a 
function $\mc{G}: [0,\infty) \to [0,\infty)$ which decreases to $0$ as $n$ tends to infinity and, for large enough $n$: 
	\begin{equation*}
		\begin{split}
\diam{\specc_{0,\La_n}(\ep)} \leq \ep \mc{G}(D_n)
		\end{split}
	\end{equation*}
Precisely, we may take:
	\begin{equation*}
		\begin{split}
\mc{G}(r) = \sum _{k=\floor{r}}^\infty \tilde{F}( \floor{k/4}) + 16 C ( M_{\Int} + \norm{\eta}_{F})   [ \Omega( k/4) + F_0( \floor{k/4})]
		\end{split}
	\end{equation*}
where $\tilde{F}$ is an $\mc{F}$-function depending on $\norm{\eta}_{F}$ and $M_\Int$.  	
\end{lemma}

\begin{proof}
Suppose $n$ is sufficiently large so that $D_n/2 >R$, the range of the interaction $\eta$. By the spectral flow decomposition in (\ref{eq:spectral-flow-decomposition}), 
	\begin{equation*}
		\begin{split}
\diam{\specc_{0,\La_n}(\ep)} & \leq 2 \norm{P\Phi^1(\ep)_0P  } \\
& = 2 \norm{ \sum \set{  P \Phi^1_x(\ep)_0 P : x\in \La_n}} \\
& \leq 2(A + B) 
		\end{split}
	\end{equation*}
for $A,B$ defined by complementary regions of the interval $\La_n = [-n, n]$:
	\begin{equation*}
		\begin{split}
A & = \sum \set{\norm{P\Phi^1_x(\ep)_0 P}~|~ \forall x\in \La_n :  \floor{ (-n+D_n)/2}  \leq  x \leq \floor{(n - D_n)/2 } }\\
B & =\sum \set{\norm{P\Phi^1_x(\ep)_0 P} ~|~ \forall x\in \La_n  : -n \leq x < \floor{(-n + D_n)/2  } \text{ or } \floor{(n - D_n)/2} < x \leq n}
		\end{split}
	\end{equation*}
By applying LTQO and $\mc{F}$-norm bounds,
	\begin{equation*}
		\begin{split}
\norm{A} \leq 8 C ( M_{\Int} + \norm{\eta}_{F}) \ep \sum _{k=\floor{D_n}}^\infty [ \Omega( k/4) + F_0( \floor{k/4})]
		\end{split}
	\end{equation*}
where $F_0$ is the shifted base $\mc{F}$-function from (\ref{eq:shift-ffunction}). For the norm bound on $B$, let $\Delta_{X(n)}$ denote the partial trace difference operators from the proof of Theorem \ref{thm:decomposition} (c.f. Theorem 6.3.4 in \cite{young:2016}), defined with respect to an enlargement of $X \subset \La_n$. Suppose 
$ -n \leq x < \floor{(-n + D_n)/2  } \text{ or } \floor{(n - D_n)/2} < x \leq n$. Denote $d_x(n) = d(x, \Int_{D_n}(\La_n))$. By the locality assumption on $\Phi^{\La_n}$ and the fact that $d_x(n)/2>R$, if $ k< d_x(n)/2 $, then, in the notation of the proof of Theorem \ref{thm:decomposition}:
	\begin{equation*}
		\begin{split}
\Phi^1(b_x(k),\ep) = \Delta _{b_x(k)} ((\alpha_\ep - id)\circ \mc{F}_{w_{\gamma_0,\ep}}(h_x) )
		\end{split}
	\end{equation*} 	
and so:
	\begin{equation*}
		\begin{split}
\norm{P \Phi^1_x(\ep)_0 P} & \leq \norm{ \sum _{k=1}^{ \floor{d_x(n)/2}} P \Phi^1(b_x(k),\ep)P} + \sum _{k= \floor{d_x(n)/2}+1}^{R_x} \norm{P \Phi^1(b_x(k),\ep)_0P} \\
& \leq \norm{ (\alpha_\ep - \mathrm{id} )\circ \mc{F}_{w_{\gamma_0,\ep}}(h_x)} +8 C ( M_{\Int} + \norm{\eta}_{F}) \ep F_0(  \floor{d_x(n)/2})
		\end{split}
	\end{equation*}
Using the quasi-locality of the generator $iD(\ep)$ of the spectral flow unitaries: 
	\begin{equation*}
		\begin{split}
(\alpha_\ep^{\La_n} - \mathrm{id}) \circ \mc{F}_{w_{\gamma_0},\ep}(h_x) & = \int _0 ^\ep i \alpha_s^{\La_n} \bigg{(} [D(s), \sum _{k=1}^{R_x} \Delta_{b_x(k)} ( \mc{F}_{w_{\gamma_0,\ep}}(h_x)] \bigg{)} ~ \mathrm{d}s
		\end{split}
	\end{equation*}
and there exists a $\mc{F}$-function $\tilde{F}$, independent of $\Lambda_n$, such that: 
	\begin{equation*}
		\begin{split}
\norm{ (\alpha-\mathrm{id})\circ \mc{F}_{w_{\gamma_0,\ep} } (h_x) } \leq \ep \tilde{F}( \floor{ d_x(n)/2})
		\end{split}
	\end{equation*}
Hence:
	\begin{equation*}
		\begin{split}
\norm{B}  \leq \sum _{k= \floor{D_n} }^\infty  \ep \tilde{F}( \floor{k/4}) +  8 C ( M_{\Int} + \norm{\eta}_{F}) \ep F_0(  \floor{k/4})
		\end{split}
	\end{equation*}
Let $\mc{G}(r) = \sum _{k=\floor{r}}^\infty  \tilde{F}( \floor{k/4}) + 16 C ( M_{\Int} + \norm{\eta}_{F_{\Phi'}})   [ \Omega( k/4) + F_0( \floor{k/4})] $. Then:
$$\diam{ \specc_{0,\La_n}(\ep)} \leq \ep \mc{G}(D_n)$$
\end{proof}

Let $P_n(\ep)$ denote the spectral projection of $H^{\partial_n}_{\Lambda_n}(\epsilon)$ associated with the isolated portion of the spectrum $\specc_{0,\La_n}(\ep)$ and define
the set of states of $\cA_{\Lambda_n}^s $, $\cS_n(\ep)$, with support in the range of $P_n(\ep)$:
$$
\cS_n(\ep) = \{ \omega \mid \omega \mbox{ is a state on } \cA_{\Lambda_n}^s \mbox{ with } \omega(P_n(\ep)) =1\}.
$$
We now consider the thermodynamic limits of these states:
$$
\cS(\ep) = \{\omega \mbox{ state on } \cA_{\Ir}^s \mid  \exists (n_k) \mbox{ increasing and } \omega_k \in \cS_{n_k}(\ep) \mbox{ s.t. } \lim_k \omega_k(A) = \omega(A), 
\forall  A \in \cA_\Ir^{\rm loc}\}.
$$

\begin{lemma} Let $c_n(\ep) = \diam{ \specc_{0,\La_n}(\ep)}$. Then\\
(i) for all $\omega \in \cS_n(\ep)$ and $A\in\cA_{\Lambda_n}^s$, we have
$$
\Re \omega(A^*[H^{\partial_n}_{\Lambda_n}(\epsilon), A]) \geq -c_n(\ep) \Vert A\Vert^2,
\mbox{ and } \left| \Im \omega(A^*[H^{\partial_n}_{\Lambda_n}(\epsilon), A])\right| \leq c_n(\ep)\Vert A\Vert^2.
$$
(ii) If $s_n=0$ and $D_{1,n}$ is such that $\lim_n[n-D_{1,n}]
=\lim_n  D_{1,n} =\infty$, then, for all $\omega \in \cS(\ep)$ and $A\in\cA^{\rm loc}_{\Ir}$, we have
$$
\lim_{n\to\infty}  \omega(A^*[H^{\partial_n}_{\Lambda_n}(\epsilon), A])\geq 0.
$$
\end{lemma}
\begin{proof}
The proof of (i) is elementary and the proof of (ii) follows by noting that the additional assumptions imply that 
the sequence $[H^{\partial_n}_{\Lambda_n}(\epsilon), A]$ converges in norm and that $\lim c_n (\ep) =0$ by Lemma \ref{lem:sublinear-locality}.
\end{proof}

In other words, the conditions of part (ii) of the lemma imply that the states in $\cS_n(\epsilon)$ converge to ground states of the infinite system. In Section 
\ref{sec:appendixTL} it is explained that the spectral flow automorphisms, like the time evolution of the system, converge to the same limit
regardless of the choice of boundary condition $\partial_n$.  Since we have the relation $P_n(0) = \alpha^{\Lambda_n,\partial_n}_\ep (P_n(\ep))$ we also have
$$
\cS_n(\ep)=\cS_n(0)\circ \alpha^{\Lambda_n,\partial_n}_\ep,
$$
and as an easy consequence of the convergence (see \cite{bachmann:2012}[Lemma 5.6]) we then also have
$$
\cS(\ep)=\cS(0)\circ \alpha_\ep.
$$

Since the same $\alpha_\epsilon$ relates limiting states regardless of the boundary conditions, for example with  constant sequence $\partial_n=\partial$, for any $n$, these
limiting states must be the same and, hence, also ground states of the infinite systems defined by the dynamics $\tau_t$. The same conclusion then holds for the lower bound
on the spectral gap above these ground states (see \cite{QLBII} for the details).

\section{Stability of spectral gap in fermion chains}

\subsection{Quasi-local maps}
Suppose $\ALa$ is a local algebra of observables which is $*$-isomorphic to $\ALa^s$. Let $\phi : \ALa \to \ALa^s$ denote a possible $*$-isomorphism. Given a local Hamiltonian $H_\La$ in $\ALa$, $\phi$ unitarily transforms $H_\La$ into a Hamiltonian $H^s_\La = \phi( H_\La)$ of the spin algebra. Using an exhaustive family of conditional expectations $\set{\theta_{X_i}: X_i \subset X_{i+1} }$, $H^s_\La$ can again be realized as the sum of local operators through a telescoping sum:
	\begin{equation*}
		\begin{split}
\forall B \in \ALa: \hspace{2mm} \phi(B) = \theta_{X_1}(\phi(B)) + \sum _{j=1} ^{N-1} \theta_{X_{j+1}} ( \phi(B)) - \theta _{X_j}(\phi(B))
		\end{split}
	\end{equation*}
The proof of Theorem \ref{thm:decomposition} uses this method of decomposition in the setting where $\phi$ is a \textit{quasi-local} $*$-automorphism, and the $\theta_{X_j}$ are normalized partial trace over increasing metric balls $X_j=b_x(j)$. The quasi-locality property, defined below, guarantees the transformed local interaction will have decay comparable to that of the original interaction. 

\medskip 

In this section, we prove stability of the spectral gap for even Hamiltonians in the CAR algebra of fermions satisfying $\zz_2$-LTQO. To do this, we will use the Jordan-Wigner isomorphism to transform even fermion interactions into spin interactions in a way that respects the parity symmetry. 

\medskip 

\begin{definition}
Let $\La\in P_f(\zz)$ be a nonempty interval. A linear map $\alpha: \mf{A}_\La^s \to \mf{A}_\La^s$ is \textit{quasi-local} if there exist constants $C>0, p\in \nn$, and decay function $g:[0,\infty) \to [0,\infty)$ such that if $X,Y\subset \La$ are disjoint subsets, then for all $A\in \mf{A}_X^s $ and $B\in \mf{A}_Y^s$, the following bounds hold: 
\begin{equation}
    \begin{split}
    \norm{\alpha(A)} \leq C|X|^{p} \norm{A} ~~~~~~~~~~~\norm{[\alpha(A),B]} \leq C\norm{A}\norm{B}|X|^{p} g(d(X,Y))
    \end{split}
\end{equation}
\end{definition}

\begin{example}
The local Heisenberg dynamics $\tau^\La : U \subseteq \rr \to \Aut(\ALa^s)$ generated by an interaction $\Psi$ with a finite $\mc{F}$-norm is a collection of quasi-local maps parametrized by $t$. Let $F$ be an $\mc{F}$-function such that $\norm{\Psi}_F < \infty$, and denote by $\nu_\Psi$ the Lieb-Robinson velocity. There exists a constant $C_\Psi>0$ such that for $X,Y \in P_f(\La)$ disjoint sets and $A\in \mf{A}_X^s, B \in \mf{A}_Y^s$, the following Lieb-Robinson bound holds: 
\begin{equation*}
\begin{split}
\norm{ [ \tau_t^\La(A), B]} \leq C_\Psi (e^{\nu_\Psi|t|}-1) \norm{A}\norm{B} \sum _{x\in X, y\in Y} F(|x-y|)  
\end{split}
\end{equation*}
But by properties of the $\mc{F}$-function: 
\begin{equation*}
\begin{split}
\sum _{x\in X, y\in Y} F(|x-y|) \leq |X| \sup \set{ \sum _{\substack{y\in \zz\\ |x-y|\geq d(X,Y) }} F(|x-y|): x\in \zz} < \infty 
\end{split}
\end{equation*}
So take $C_{t} = C_\Psi (e^{\nu_\Psi|t|}-1)$, $p_{t}=1$ and: $$g_{t}(n) = \sup \set{ \sum_{\substack{y\in \zz\\ |x-y|\geq n}} F(|x-y|):x\in \zz}.$$ In particular, the spectral flow automorphism $\alpha^\La : [0, \ep_\La]\to \Aut(\ALa^s)$ is quasi-local \cite{bachmann:2012}.  
\end{example}
Lastly, we specify the normalized partial trace maps. Let: $$X(n)=\set{ z\in \La: \exists x\in X, ~ |z-x|\leq n}$$ 
\ldots denote an enlargement of $X\in P_f(\La)$. Denote the normalized partial trace of the state space over $\La\setminus X(n)$ by: 
\begin{equation*}
	\begin{split}
	\theta_{X(n)} = \frac{1}{\dim \mf{H}_{\La \setminus X(n)}} \trr_{\mf{H}_{\La 	\setminus X(n)}}.
	\end{split}
\end{equation*}
For convention, we will take the trace over $\mf{H}_\emptyset$ as the identity map. Then define, for all $A \in \mf{A}^s_\La$:
\begin{equation*}
\begin{split}
 \Delta_{X(0)} (A) = \theta_{X(0)}(A), ~~ \Delta_{X(n)}(A) = \theta_{X(n)}(A) - \theta_{X(n-1)}(A).
\end{split}
\end{equation*} 

\subsection{Transformation of even fermion interactions} 

Recall, we denote by $\A^+_\La \subset \A^f_\La$ the even operators of the CAR algebra over $\La$. We say $\beta \in \text{Aut}(\A_\La)$ is even if it preserves the parity. Even interactions are defined similarly. We also denote $S^{\pm}= \frac{1}{2}( \sig^1 \pm i \sig^2)$. The following definition is the well-known Jordan-Wigner transformation, which gives a $C^*$-isomorphism of CAR and spin algebras. 

\medskip 

\begin{definition}
Consider the case $\mf{A}_{\set{x}}^s = M_2(\cc)$. Let $\vtheta_\La: \mf{A}_\La^f \to \mf{A}^s_\La$ denote the \textit{Jordan-Wigner map} defined by: 
	\begin{equation*}
		\begin{split}
a(x) \mapsto \exp\bigg{(} -i\pi \sum _{j<x} S^+_j S^-_j \bigg{)} S_x^- \hspace{10mm} a^*(x) \mapsto \exp \bigg{(} i\pi \sum _{j<x} S_j^+ S_j^- \bigg{)} S_x^+
		\end{split}
	\end{equation*}
\end{definition}
The Jordan-Wigner transformation extends the notion of parity to the spin $1/2$ algebra. We say $A \in \A^s_\La$ is even if $\vartheta_\La^{-1}(A) \in \A^+_\La$. 
\begin{proposition}\label{lem:support}
Let $X \subset \La$ be any subinterval. 
	\begin{equation*}
		\begin{split}
& 1. ~ \text{If }A \in \A^+_X, \text{ then }\vtheta_\La(A) \in \A^s_X \\
& 2. ~ \text{If }\alpha:\A_\La^s\to \A_\La^s \text{ is an even quasi-local map, then }\Delta_{X(n)}\circ \alpha \text{ is also even.}  
		\end{split}
	\end{equation*}
\end{proposition}

\begin{proof}
Suppose $A$ is a monomial $ca^\#(x_1)\cdots a^\#(x_{2n})$. By the CAR, we may assume $x_j \leq x_{j+1}$. A direct computation shows the first part of the lemma holds for the even monomials which generate $\A^+_X$: 
\begin{equation*}
	\begin{split}
\vtheta_\La(A) = c \prod _{k=2}^{2n} S^\flat _{x_{k}} S^\flat_{x_{k-1}} \exp\bigg{(} \pm i\pi \sum _{j=x_{k-1}}^{x_{k}-1}S^+_j S^-_j      \bigg{)}  \in \mf{A}_X^s. 
	\end{split} 
\end{equation*}
Next, we show that the partial trace is an even map. For any $x\in \La$, define the following four unitary operators:
	\begin{equation*}
		\begin{split}
u^{(0)}_x = \1,~ u^{(1)}_x = \sig_x^1, ~u^{(2)}_x = i\sig_x^2,~ u^{(3)}_x = \sig^3_x 
		\end{split}
	\end{equation*}
Now, let $Z \subset \La$ and $B \otimes C \in \A^s_Z \otimes \A^s_{\La \setminus Z}$. Denote by $I_{\La \setminus Z}$ the set of finite sequences $\iota: \La \setminus Z \to \set{0,1,2,3}$. Define: 
 	\begin{equation*}
		\begin{split}
u(\iota) = \prod _{z\in Z} u_z^{(\iota_z)}. 
		\end{split}
	\end{equation*}
Using elementary properties of trace and locality in the spin algebra: 
	\begin{equation}\label{eq:partialtrace}
		\begin{split}
\frac{1}{\dim(\mf{H}_{\La \setminus Z})} B\otimes \tr{C}\1 = \frac{1}{4^{|\La \setminus Z|}} \sum _{\iota \in I_{\La \setminus Z}} u(\iota)^* [B\otimes C ]u(\iota) \in \A^s_Z
		\end{split}
	\end{equation}
The relation in (\ref{eq:partialtrace}) uniquely defines the partial trace, hence: 
	\begin{equation*}
		\begin{split}
\theta_Z(\cdot) = \frac{1}{4^{|\La \setminus Z|}} \sum _{\iota \in I_{\La \setminus Z}} u(\iota)^* [\cdot] u(\iota). 
		\end{split}
	\end{equation*}
The second part of the lemma follows from this formula. 
\end{proof}

In the following, we will assume the interactions are supported on intervals: 

\begin{definition}
An interaction $\Phi$ is \textit{supported on intervals} if $\Phi(X) \not = 0$ only if $X = [a,b]$ for some $a,b\in \zz$. 
\end{definition}

Any interaction can be ``regrouped" into one with interval support, and while the methods to do this are neither new nor canonical, we record here a simple way without changing the local Hamiltonians, at the expense of rate of decay. 

\begin{proposition}\label{prop:interval-support}
Suppose $I\subset \zz$ is an interval and  $\Psi: P_f(I)\to \Aloc$ is an interaction. Then there exists an interaction $\Phi : P_f(I) \to \Aloc$, supported on intervals, such that for all finite intervals $\La \subset I$, the local Hamiltonians are equal:
	\begin{equation*}
		\begin{split}
\Phi_\La = \sum_{X\subset \La} \Phi(X) = \Psi_\La
		\end{split}
	\end{equation*}  
If $\Psi$ is an unperturbed interaction with uniform bound $M$, range $R$, and local gap $\gamma_0$, then so is $\Phi$, with uniform bound $2^{R}M$ and the same range and local gap. 

\medskip 

Furthermore, if $\norm{\Psi}_F<\infty$, where $F$ is the $\mc{F}$-function in (\ref{eq:ffunction}), and $h(r) \geq Kr^s$ for some $K>0$ and $s\in (0,1]$, then $\norm{\Phi}_{G}\leq \norm{\Psi}_F$ for the $\mc{F}$-function defined:
	\begin{equation*}
		\begin{split}
G(r) & = e^{-\frac{1}{2}h(r)} \frac{C_\Phi}{(1+cr)^\kappa}, \\
C_\Phi & = L\sum _{n=1}^{\infty} n e^{- \frac{1}{2} h(n)}
		\end{split}
	\end{equation*}
\end{proposition}

\begin{proof}
If $I \subsetneq \zz$, then we may extend $\Psi$ to $\zz$ by $\Psi(Z)=0$ for $Z \not \subset I$, and by construction, $\Phi$ defined in terms of the extension will restrict to an interaction on $I$. So we may assume $I=\zz$. We will define $\Phi$ by induction on the diameter $n$ of intervals $[k,k+n]$. When $n=0,1$ define: 
	\begin{equation*}
		\begin{split}
\Phi(\set{x}) = \Psi(\set{x}) \text{ and }\Phi( \set{x,x+1}) = \Psi(\set{x,x+1})
		\end{split}
	\end{equation*}
For larger $n$, define: 
	\begin{equation*}
		\begin{split}
\Phi( [k,k+n]) = \sum \set{ \Psi(X): X\subset [k,k+n], \diam{X}=n}
		\end{split}
	\end{equation*}	
By construction, $\Phi_\La = \Psi_\La$. Now, suppose $\Phi$ is an unperturbed interaction with constants $M,R, \gamma_0$. Since $\Phi_{b_\La(x,n)}=\Psi_{b_\La(x,n)}$ for all $x$ and $n$, $\Phi$ and $\Psi$ have the same local gap. Similarly, it is clear that $\Phi$ and $\Psi$ have the same range $R$, and if $\diam{[a,b]} \leq R$: 
	\begin{equation*}
		\begin{split}
\norm{ \Phi( [a,b]) } \leq 2^{R} M
		\end{split}
	\end{equation*}
Now, suppose $\Phi$ is some interaction, not necessarily finite range, with $\norm{\Phi}_F$. For fixed $k\in \zz$ and $n\geq 0$, by Proposition \ref{prop:FfunctionProperties}:
	\begin{equation*}
		\begin{split}
\norm{ \Phi( [k,k+n]) } \leq \sum _{\substack{X\in P_f(\zz)\\ k,k+n\in X}} \norm{\Psi(X)} \leq \norm{\Psi}_F F(n) 
		\end{split}
	\end{equation*}
So for $x,y\in \zz$: 
	\begin{equation*}
		\begin{split}
\sum_{\substack{k,n\\ x,y\in [k,k+n]}} \norm{ \Phi( [k,k+n])} & = \sum _{n\geq |x-y|} \sum _{\substack{k\\ x,y \in [k,k+n]}} \norm{ \Phi([k,k+n])} \\ 
& \leq \sum _{n\geq |x-y|} \sum _{\substack{k\\ x,y\in [k,k+n]}} \norm{\Psi}_F F(n) \\
&\leq  \norm{\Psi}_F\sum_{n \geq |x-y|} (n+1-|x-y|) e^{-h(n)} \frac{L}{(1+cn)^\kappa} \\
& \leq\norm{\Psi}_F\bigg{(}\sum_{n=1}^\infty ne^{-\frac{1}{2}h(n)} \bigg{)} e^{-\frac{1}{2}h(|x-y|)} \frac{L}{(1+c|x-y|)^\kappa}
		\end{split}
	\end{equation*}
That is, 
	\begin{equation*}
		\begin{split}
\norm{\Phi}_G= \sup_{x,y\in \zz}\set{ \sum _{\substack{X\in P_f(\zz)\\ x,y\in X}} \frac{\norm{\Phi(X)}}{G(|x-y|)} }& \leq \norm{\Psi}_F
		\end{split}
	\end{equation*}
\end{proof}

\begin{proposition}\label{prop:fermion-to-spin}
Suppose $\Psi:P_f(I)\to \Aloc^f$ is an even interaction supported on intervals. Then there exists an even interaction $\Phi: P_f(I)\to \Aloc^s $ such that for any $\La \subset I$: 
	\begin{equation*}
		\begin{split}
\vtheta_\La(\Psi_\La) =  \Phi_\La
		\end{split}
	\end{equation*}
If $\Psi$ satisfies a finite $\mc{F}$-norm for some $F$ of the form (\ref{eq:ffunction}), then so does $\Phi$. If $\Psi$ is an unperturbed interaction, then so is $\Phi$ for the same constants. 
\end{proposition}

\begin{proof}
For $\La_0 \subset \La$, let $\iota_{\La_0, \La }$ denote the inclusion $\A^f_{\La_0} \hookrightarrow \A^f_\La$. If $A \in \A_{\La_0}^+$, then by expanding in an even generating set of monomials we see:
	\begin{equation*}
		\begin{split}
\iota_{\La_0, \La} \circ \vartheta_{\La_0}(A) = \vartheta_{\La} \circ \iota_{\La_0, \La}(A)
		\end{split}
	\end{equation*}
So there exists an injective $*$-morphism $\vartheta:\bigcup \A_\La^+ \to \Aloc^s$ which extends every $\vartheta_\La$, from which we define $\Phi(X) = \JW{\Psi(X)}$. By Proposition \ref{lem:support}, this is a well-defined interaction which is also supported on intervals. Evidently $\Phi$ is an even interaction, i.e. $\vartheta^{-1}(\Phi(X)) $ is even for any $X$. 

$\vartheta$ is isometric, and for the $\mc{F}$-function $F$: 
	\begin{equation*}
		\begin{split}
\norm{\Psi}_F = \sup _{x,y} \sum _{\substack{X\in P_f(\zz)\\ x,y\in X}} \frac{\norm{\vartheta(\Psi(X))}}{F(|x-y|)} = \norm{\Phi}_F
		\end{split}
	\end{equation*}
Now suppose $\Psi$ is an unperturbed interaction. Then evidently $\Phi$ is uniformly bounded. $\Phi$ is frustration free and uniformly locally gapped since, for any $\La$, there exists a unitary $Q_\La : \mf{H}_\La \to \mf{F}_\La$ such that for $A\in \A_\La^f$:
	\begin{equation*}
		\begin{split}
\vartheta(A) = \vartheta_\La(A) = Q^*_\La A Q_\La
		\end{split}
	\end{equation*}
Since $\vartheta$ is an isometry which preserves support for even observables, and $Q_\La$ is unitary, $\Phi$ has the same uniform bound, range, and local gap as $\Psi$.  
\end{proof}

\begin{theorem}[Ground state gap stability for fermion chains]\label{thm:fermion-gap-stability}
There exist $\ep_{\gamma_0}'>0$ and constant $m_D'$ such that $0\leq \ep < \ep_{\gamma_0}'$ and $\diam{\La}>\max\set{2D,R}$ implies: 
	\begin{equation*}
		\begin{split}
\gamma(H_\La(\ep)) \geq \gamma_0 - m_D'\ep > 0
		\end{split}
	\end{equation*}
The constants $m_D'$ and $\ep_{\gamma_0}'$ can be explicitly determined by the expressions in (\ref{eq:specific-constants}). 
\end{theorem}

\begin{proof}
By Proposition \ref{prop:interval-support}, we assume $\eta$ and $\Phi^\La = \Phi $ are supported on intervals. Proposition \ref{prop:fermion-to-spin} implies the existence of spin interactions $\eta^S$ and $\Phi^S$ with the same uniform bound, range, local gap $\gamma_0$ and decay. 

\medskip 

Let $\gamma \in (0,\gamma_0)$ and $D\in \nn$ be a chosen distance from the boundary, uniform in the volume, and consider fixed $\La$ with sufficiently large diameter. By Theorem \ref{thm:decomposition}, the spectral flow decomposes the local Hamiltonian $H_\La(\ep)$ of $\eta^S + \ep \Phi^S$: 
	\begin{equation*}
		\begin{split}
\alpha_\ep^\La(H_\La+ \ep \Phi_\La) & = H_\La + \sum_{x\in \La} \Phi^1_x(\ep) = H_\La + \Phi^2(\ep) + \Phi^3(\ep) + \mc{R}(\ep) + \omega_\La( \widetilde{\Phi^1(\ep)})
		\end{split}
	\end{equation*}

Since $\vartheta_\La$ is implemented by some unitary, $\eta^S$ has $\zz_2$-LTQO for the same decay function $\Omega$. So to apply the norm boundedness argument in Section 3, it suffices to argue that $\Phi^1(b_\La(x,n),\ep)$ is even. 

\medskip 

But the proof of Theorem \ref{thm:decomposition} in \cite{young:2016} guarantees the existence of even interactions $\Psi_i: P_f(\La)\to \A_\La^s$, $i=1,2,3$, and quasi-local maps $\mc{K}_{i}^{(\ep)}: \A_\La^s \to \A_\La^s$ such that: 
	\begin{equation*}
		\begin{split}
\Phi^1(b_\La(x,n),\ep) = \Delta_{b_\La(x,n)} \circ \mc{K}_1^{(\ep)} &( \Psi_1(\set{x})) +  \sum_{k=1}^n \ep \Delta_{b_\La(x,n)}\circ \mc{K}_2^{(\ep)} ( \Psi_2(b_\La(x,k)))  \\
&  + \Delta_{b_\La(x,n)} \circ \mc{K}_3^{(\ep)} (\Psi_3(b_\La(x,k)))
		\end{split}
	\end{equation*}
The $K_i^{(\ep)}$ are defined in terms of the spectral flow automorphism and are also even maps. Hence, by Lemma \ref{lem:support}, $\Phi^1(b_\La(x,n),\ep)$ must also be even, since the even observables form a subalgebra.   
\end{proof}

\section{Example of even Hamiltonian satisfying stability hypotheses}

Here we describe an example of an interaction of the CAR algebra which satisfies the stability hypotheses of Theorem \ref{thm:spin-gap-stability}. Let $\mc{X}=\set{ f_i : i \in \mc{B}}$ and $\mc{Y}=\set{g_j: j \in \mc{B}}$ be two collections of vectors in $\ell^2(\zz)$ such that: 

\medskip

(i) $\mc{X}\cup \mc{Y}$ is an orthonormal basis for $\ell^2(\zz)$. 

\medskip 

(ii) There exist $R\geq 0$ and collections $\set{x_i:i\in \mc{B}}$, $\set{y_j: j\in \mc{B}}$ such that for all $i,j$: 
\begin{gather*}
\text{supp}(f_i) \subset b(x_i, R) \hspace{5mm} \text{supp}(g_j) \subset b(y_j, R) \\
i \not = j \text{ implies } b(x_i, R) \cap b(x_j,R) = \emptyset = b(y_i,R)\cap b(y_j,R)
\end{gather*}

\noindent Furthermore, denote: $\mc{X}_W  = \set{ f_i : \text{supp}(f_i) \subset W}$ and $\mc{Y}_W = \set{ g_j : \text{supp}(g_j) \subset W}$. We will also assume: 
\medskip 

(iii) There exists a diameter $N_0$ such that for all intervals $\La$, $\diam{\La}>N_0$ implies $\mc{X}_\La \not = \emptyset$ and $\mc{Y}_\La \not = \emptyset $.

\begin{definition}
Let $\eta : P_f(\zz) \to \Aloc^f$ be the finite-range interaction defined by: 
\begin{equation}\label{eq:orbitals}
\begin{split}
\eta( b(x_i, R)) = \1 - a^*(f_i)a(f_i) \hspace{20mm} \eta(b(y_j,R))=a^*(g_j)a(g_j)
\end{split}
\end{equation} 
\end{definition}

\begin{lemma}
Suppose $\La$ is an interval such that $diam(\La)>N_0$. Then $H_\La $ is non-negative, uniformly gapped and frustration free. 
\end{lemma}

\begin{proof}
Let $(f_{n_1},\ldots,f_{n_\La})$ and $(g_{m_1},\ldots,g_{m_\La})$ be the collections of vectors whose support is contained in $\La$. If necessary, complete the list to an orthonormal basis of $\ell^2(\La)$ with $(h_1,\ldots, h_{p})$, $p= |\La| - n_\La - m_\La$. Evidently $H_\La$ is uniformly gapped and non-negative. So we prove that:
\begin{equation*}
\begin{split}
\ker(H_\La) = \text{span}( \psi_X : X \subset [1,p] )
\end{split}
\end{equation*}
where we define: 
\begin{equation*}
\begin{split}
\phi_\La =  f_{n_1}\wedge \cdots f_{n_\La}, ~~~ \zeta_X = \bigwedge \set{ h_{i_k}: i_k \in X \subset [1,p]},~~~\psi_X =  \phi_\La \wedge \zeta_X
\end{split}
\end{equation*}
By calculation, $\psi_X \in \ker(H_\La)$ for any $X \subset [1,p]$. But each term of the interaction $H_\La$ is a projection, the complement projection of some $a^*(f_q)a(f_q)$. So $H_\La \psi =0$ implies $\psi \in \text{ran}(a^*(f_{i}) a(f_{i}))$ for each $i=n_1,\ldots, n_\La$. 
\end{proof}

Next, we show that the number of auxiliary orthonormal basis vectors $h_i$ needed to complete $\mc{X}_\La$ and $\mc{Y}_\La$ to a basis of $\ell^2(\La)$ is uniformly bounded in $\La$, and that each $h_i$ has support contained towards the edge of $\La$. 

\begin{lemma}
Suppose $\diam{\La}>N_0$. Let $\mc{Z}(\La) = \set{ h_1^{(\La)}, \ldots, h_n^{(\La)}}$, $n=n(\La)$, be a basis for the complement of $span( \mc{X}_\La \cup \mc{Y}_\La)$ in $\ell^2(\La)$. Then: 
\begin{enumerate}
\item For each $i\in [1,n]$,  \emph{$\text{supp}(h_i^{(\La)}) \subset \La \setminus \text{Int}_{3R}(\La) $}
\item $ |\mc{Z}(\La)| \leq 6 R $
\end{enumerate} 
\end{lemma}

\begin{proof}
Let $(\xi_k)$ denote the orthonormal basis from $\mc{X}\cup \mc{Y}$. Suppose $\text{supp}(f) \subset \text{Int}_{3R}(\La)$. Then $x_i \not \in \La$ implies $\mathrm{\text{supp}}(f_i) \cap \text{supp}(f) = \emptyset$, that is, $\ip{f_i,f}=0$ (resp. $y_j$ and $\ip{g_j,f}=0$). Hence: 
\begin{equation*}
\begin{split}
f = \sum \ip{\xi_k,f} \xi_k = \sum _{ \xi \in \mc{X}_\La \cup \mc{Y}_\La} \ip{\xi,f}\xi 
\end{split}
\end{equation*} 
Hence $f \in \text{span}( \mc{X}_\La \cup \mc{Y}_\La)$. Now, a basis of the orthogonal complement of $\ell^2( \text{Int}_{3R}(\La))$ in $\ell^2(\La)$ is necessarily supported on $\La \setminus \text{Int}_{3R}(\La)$, proving (1).  Additionally, the dimension of $\ell^2( \La \setminus \text{Int}_{3R}(\La))$ is an upper bound for $|Z(\La)|$, which proves (2).  
\end{proof}

This lemma has an immediate corollary: 

\begin{corollary}\label{cor:interior}
Let $\mf{A}(\mc{W})$ denote the $C^*$-subalgebra of $\mf{A}_\zz^f$ generated by the operators $a^*(f),a(f)$ such that $f\in \mc{W} \subset \ell^2(W)$. Then for all intervals $\La$ with diameter larger than $6R$: 
\begin{equation*}
\begin{split}
\mf{A}_{\emph{Int}_{3R}(\La)} \subset \mf{A}(\mc{X}_\La\cup  \mc{Y}_\La)
\end{split}
\end{equation*}  
\end{corollary}


To conclude this section, we prove that the interaction defined in (\ref{eq:orbitals}) satisfies $\zz_2$-LTQO. Denote $D = \max\set{ N_0, 3R}$. Recall that if $n\geq D$ then $H_{b_\La(x,n)}$ is non-negative and frustration free with kernel indexed by $\mc{Z}(b_\La(x,n))$. 

\medskip 

Define the step-function $\Omega:[0,\infty) \to [0,\infty)$ by: 
\begin{equation*}
\begin{split}
\Omega(x) = \bigg{ \{ } \begin{array}{l l}0 & \text{ if } x \geq D \\ 2 & \text{ otherwise} \end{array}
\end{split}
\end{equation*}

\begin{proposition}\label{prop:orbitals-LTQO}
Suppose $\diam{\La}>2D$, and let $x\in \La$, and $(n,k)\in \nn^2$ be such that $0\leq k \leq r_x$, $k \leq n \leq R_x$. Let $P_n$ denote the projection onto $H_{b_\La(x,n)}$. Then for all $A \in \mf{A}_{b_\La(x,k)}^+$: 
\begin{equation*}
\begin{split}
\norm{P_n (A-\omega_\La(A))P_n} \leq \Omega(z_x(n)-k) \norm{A}
\end{split}
\end{equation*}
\end{proposition}

\begin{proof}
We handle the two cases of $n$ when $\diam{b_\La(x,n)} \geq N_0$ or $\diam{b_\La(x,n)} < N_0$. Suppose the former. Now, there are two subcases for $k$: either $b_\La(x,k) \not \subset \text{Int}_D(b_\La(x,n))$ or $b_\La(x,k)$ is contained in that interior. 

\medskip 

Suppose $b_\La(x,k) \subset \text{Int}_D(b_\La(x,n))$. Then $z_x(n)-k\geq D$, necessarily. Denote:
\begin{equation*}
\begin{split}
\mc{X}_{b_\La(x,n)} = \mc{X}_n= \set{ f_{i_1}\ldots, f_{i_M} } ~~~ \mc{Z}(b_\La(x,n)) =\mc{Z}(n)= \set{ h_1, \ldots , h_p }
\end{split}
\end{equation*} 
Let $\psi^n_X = f_{i_1} \wedge \cdots f_{i_M} \wedge h_{n_1} \wedge \cdots h_{n_{|X|}}$ be a generic unit norm basis vector of the kernel, indexed by $X \subset \mc{Z}(n)$. A calculation shows:
\begin{equation*}
\begin{split}
\norm{ P_n (A - \omega_\La(A))P_n} \leq 6R \sup _{X \subset \mc{Z}(n)} | \ip{\psi^n_X, A \psi^n_X} - \omega_\La(A)| + 2^{6R} \sup _{X\not = Y} |\ip{\psi_X^n, A\psi_Y^n} |
\end{split}
\end{equation*}
But by the theory of quasi-free states and Corollary 5.3: 
	\begin{equation*}
		\begin{split}
\sup _{X \subset \mc{Z}(n)} | \ip{ \psi_X^n, A \psi_X^n} - \omega_\La(A)| = \sup _{X\not = Y} |\ip{ \psi^n_X, A \psi^n _Y }|=0
		\end{split}
	\end{equation*}
Now suppose $b_\La(x,k)$ is not contained in the $D$-interior of $b_\La(x,n)$. This implies $z_x(n)-k < D$. And by the trivial commutator bound: 
\begin{equation*}
\begin{split}
\norm{P_n(A-\omega_\La(A))P_n} \leq 2 \norm{A} = \Omega(z_x(n)-k)\norm{A}
\end{split}
\end{equation*}
Lastly, suppose $\diam{b_\La(x,n)} < N_0$. Then $n-k \leq n < N_0\leq D$. Hence $z_x(n)-k < D$ as well, and the trivial bound agrees with $\Omega$. Conclude that $H_\La$ satisfies LTQO for $\Omega.$  
\end{proof}

\section{Appendix}

\subsection{$\mc{F}$-functions and decay of interactions} 

In addition to LTQO, a critical assumption for our spectral gap stability argument is rapid decay of the perturbations in $[\Phi]$. We choose to describe this decay through $\mc{F}$-functions, which have several useful properties, one of which is defining an extended norm on the real vector space of interactions. 

\begin{definition}
A function $F: [0,\infty) \to (0,\infty)$ is an $\mc{F}$-\textit{function for $(\zz,|\cdot|)$} if: 
\begin{enumerate}
\item $\norm{F}= \sum_{x\in \zz} F(x) < \infty $ ; 
\item there exists $C_F>0$ such that for all $x,y\in \zz$: 
	\begin{equation*}
		\begin{split}
\sum _{z\in \zz} F(|x-z|) F(|z-y|) \leq C_F F(|x-y|).
		\end{split}
	\end{equation*}
\end{enumerate}
Furthermore, if $\Phi$ is an interaction, then the \textit{$\mc{F}$-norm} of $\Phi$ (with respect to $F$) is defined: 
	\begin{equation}\label{eq:fnorm}
		\begin{split}
\norm{\Phi}_F = \sup \set{ \sum _{\substack{Z \in P_f(\zz)\\x,y\in Z}} \frac{\norm{\Phi(Z)}}{F(|x-y|)} } \in [0,\infty]. 
		\end{split}
	\end{equation} 
\end{definition} 

\begin{example}
Suppose $h:[0,\infty) \to [0,\infty)$ is a monotone increasing, subadditive function and $\kappa>2$. The following function $F$ defines an $\mc{F}$-function:
	\begin{equation}\label{eq:ffunction}
		\begin{split}
F(r) = e^{-h(r)}\frac{L}{(1+cr)^\kappa}, ~L,c>0
		\end{split}
	\end{equation} 
\end{example}
The $\mc{F}$-function in (\ref{eq:ffunction}) and following properties will be used extensively in the proof of spectral gap stability. 

\begin{proposition}\label{prop:FfunctionProperties}
Suppose $\Phi$ is an interaction with finite $\mc{F}$-norm for some $F$. Then: 
	\begin{equation*}
		\begin{split}
& 1. ~ \text{For any collection $Z_1 \subset Z_2 \subset \cdots \subset Z_N$, $\sum _{k=1}^{N} \norm{\Phi(Z_k)} \leq \norm{\Phi}_F F( {\diam{Z_1}})$ }\\ 
& 2. ~ \text{If $\eta$ is a uniformly bounded, finite range interaction, then $\norm{\eta}_F< \infty$ }\\ & ~~~\text{and $\norm{\eta+\Phi}_F < \infty$. }
		\end{split}
	\end{equation*} 
\end{proposition}

\begin{proof} Let $\diam{Z_1}=n$, and choose $x,y\in Z_1$ such that $|x-y|=n$. Then:
	\begin{equation*}
		\begin{split}
\sum_{k=1}^{N} \norm{ \Phi(Z_k)} \leq \sum _{\substack{X\in P_f(\zz)\\ x,y\in X}} \norm{\Phi(X)} \leq \norm{\Phi}_F F(n)
		\end{split}
	\end{equation*}
Now, denote the range of $\eta$ by $R$ and uniform bound by $M$. Suppose $x,y\in \zz$. If $Z\in P_f(\zz)$ contains $x,y$ and $\Phi(Z)\not = 0$, then $Z \subset b(x,R)\cap b(y,R)$. Hence: 
	\begin{equation*}
		\begin{split}
\sum _{\substack{X\in P_f(\zz)\\ x,y\in X} }\norm{\eta(X)} \leq 2^{3R} M 
		\end{split}
	\end{equation*}
Then $\norm{\eta + \Phi}_F < \infty$ by the triangle inequality.  
\end{proof} 

\subsection{Thermodynamic limit of the spectral flow}\label{sec:appendixTL}

There are standard results giving conditions on an interaction $\Phi$ under which the finite-volume dynamics defined by 
$$
\tau_{s,t}^{\Phi,\Lambda}(A) = U(s,t)^* A U(s,t), A\in \cA_\Lambda,  
$$
with 
$$
H_\Lambda(s)=\sum_{X\subset \Lambda} \Phi(X,s),
$$
and
$$
\frac{d}{ds} U(s,t) = i H_\Lambda(s) U(s,t), U(t,t)=\idty, s,t \in I\subset \Rl,
$$
converges to a strongly continuous co-cycle of automorphism, $\tau^\Phi_{s,t}$, of the algebra of quasi-local observables $\cA_\Gamma$, which is defined as the norm completion of the (strictly) local
observables given by
$$
\cA^{\rm loc}_\Gamma = \bigcup_{\Lambda\in\cP_f(\Gamma)} \cA_\Lambda .
$$
One sufficient condition is that $\Phi(\cdot,t)$ is a continuous curve taking values in the space $\cB_F$, where $F$ defines an $F$-norm, $\Vert\cdot\Vert_F$
on interactions as in the previous section of this Appendix \cite{nachtergaele:2006}. For any compact interval $I$, we can define the space $\cB_F(I)$ as the 
set of all such continuous curves and for $\Phi\in\cB_F(I)$ the function
\begin{equation}
\Vert\Phi \|_F(t) = \sup_{x,y \in \Gamma} \frac{1}{F(d(x,y))} \sum_{\stackrel{X \in \mathcal{P}_0( \Gamma):}{x,y \in X}} \| \Phi(X, \cdot) \|(t) 
\end{equation}
is continuous and bounded. Strong convergence of the automorphisms here means that 
$$
\lim_{\Lambda\uparrow\Gamma} \Vert \tau_{s,t}^{\Phi,\Lambda}(A) - \tau_{s,t}^\Phi(A)\Vert =0, \mbox{ for all } A\in \cA^{\rm loc}_\Gamma.
$$
The limit is taken over any sequence of $\Lambda_n\in\cP(\Gamma)$ increasing to $\Gamma$ and the limiting dynamics is independent of the choice of sequence.

One can also show that the dynamics depends continuously on the interaction $\Phi$ in the following sense \cite{QLBI}:
\begin{equation}\label{cont_bd_iv} 
\| \tau_{t,s}^\Phi(A) - \tau^\Psi_{t,s}(A) \| \leq \frac{2 \| A \|}{C_F} \Vert F\Vert |X| e^{2\min(I_{t,s}(\Phi),I_{t,s}(\Psi))} I_{t,s}( \Phi - \Psi).  
\end{equation} 
holds for all $A \in \mathcal{A}_X$ and $s,t \in I$, and where
where, for $\Phi\in\cB_F(I)$, and $s,t\in I$, the quantity $I_{t,s}(\Phi)$  defined by
\be
I_{t,s}(\Phi) = C_F\int_{\min(t,s)}^{\max(t,s)} \| \Phi \|_F(r) \, dr,
\label{ItsPhi}\ee

It is often important to include in the definition of the finite-volume Hamiltonians $H_\Lambda$ terms that correspond to a particular boundary condition. 
Such terms affect the ground states and equilibrium states of the system, including in the thermodynamic limit but, in general, do not affect the infinite-volume dynamics. 
In order to express this freedom in the interactions defining the finite-volume dynamics that lead to the same thermodynamic limit, we use another, weaker,
notion of convergence of interactions interaction $\Phi$ introduced in \cite{QLBI}, where it is called {\em  local convergence in $F$-norm}.

\begin{definition} \label{def:lcfnorm} Let $(\Gamma, d, \, \mathcal{A}_{\Gamma})$ be a quantum lattice system, $F$ an $F$-function for $(\Gamma, d)$, and $I \subset \mathbb{R}$ an interval. We say that a sequence of interactions $\{ \Phi_n \}_{n \geq 1}$ \emph{converges locally in $F$-norm} to $\Phi$  such that:
\begin{enumerate}
	\item[(i)]$\Phi_n \in \mathcal{B}_F(I)$ for all $n \geq 1$,
	\item[(ii)] $\Phi \in \mathcal{B}_F(I)$,
	\item[(iii)] For any $\Lambda \in \mathcal{P}_0( \Gamma)$ and each $[a,b] \subset I$, one has that
	\begin{equation}
	\lim_{n \to \infty} \int_a^b \| ( \Phi_n - \Phi) \restriction_{\Lambda} \|_F(t) \, dt  = 0 \, . 
	\end{equation}
\end{enumerate}
\end{definition}

In this appendix, we want to apply this notion to the spectral flow generated by perturbations of the form \eq{eq:deform} and its thermodynamic limit.


The spectral flow $\alpha_\epsilon^{\Lambda,\partial}$ for the curve of Hamiltonians $H^\partial_\Lambda(\epsilon), \epsilon\in[0,\epsilon_0)$,
defined in \eq{curveham}, also depends on a parameter $\gamma>0$.
This parameter is assumed to be a lower bound for the gap of interest in the spectrum of $H^\partial_\Lambda(\epsilon)$ in the stability argument, but this assumption is not needed for
the construction of $\alpha_\epsilon^{\Lambda,\partial}$. The automorphisms $\alpha_\epsilon^{\Lambda,\partial}$ are generated by the self-adjoint operators $D_\Lambda^\partial(\epsilon)$,
defined by
\be
D_\Lambda^\partial(\epsilon) = \cK^{\Lambda,\partial}_\epsilon\left( \sum_{X\subset\Lambda_{D_1}}\Phi(X) + s \sum_{X\subset(\Lambda\setminus\Lambda_{D_2})}\Phi(X)\right),
\ee
where the map $\cK^{\Lambda,\partial}_\epsilon:\cA_\Lambda\to\cA_\Lambda$, is given by
\be
\cK^{\Lambda,\partial}_\epsilon(A) = \int_{-\infty}^\infty \tau^{H^\partial_\Lambda(\epsilon)}(A)W_\gamma(t) dt.
\ee
Note that $\cK^{\Lambda,\partial}_\epsilon$ is defined as a linear map but it depends itself on $H^\partial_\Lambda(\epsilon)$ and therefore $D_\Lambda^\partial(\epsilon)$
depends non-linearly on the perturbation. In \cite{QLBI}[Section 5.4] a detailed study of transformations of the form $\cK^{\Lambda,\partial}_\epsilon$ is performed. The following 
proposition follows directly from applying the more general results in that work to the situation here.

\begin{proposition}[\cite{QLBI}]\label{prop:loc_conv_in_fnorm}
There exists an $F$-function $\tilde F$ of the form \eq{eq:ffunction} and interactions $\Psi^{\Lambda,\partial}\in\cB_{\tilde F} ([0,\epsilon_0])$ such that 
$$
D_\Lambda^\partial(\epsilon) =  \sum_{X\subset \Lambda} \Psi^{\Lambda,\partial}(X,\epsilon).
$$ 
Furthermore, there exists an interaction $\Psi\in\cB_{\tilde F} ([0,\epsilon_0])$ such that for any sequences $\Lambda_n=[a_n,b_n]\subset \Ir$, $\partial_n=(D_{1,n} , D_{2,n},s_n)$ , 
the interactions $\Psi^{\Lambda_n,\partial_n}$ have uniformly bounded $\tilde F$-norm and converge locally in 
$\tilde F$-norm to $\Psi$.
\end{proposition}

As a consequence, we can apply the following theorem from \cite{QLBI} to the sequences of interactions $\Psi^{\Lambda_n,\partial_n}$.

\begin{theorem}[\cite{QLBI}[Theorem 3.8] \label{thm:exist_iv_dyn}
Let $(\Phi_n )_{n \geq 1}$ be a sequence of time-dependent interactions on $\Gamma$ with
$\Phi_n$ converging locally in the $F$-norm to $\Phi$ with respect to $F$. Suppose that for every $[a,b] \subset I$,
\be \label{sup_bd}
\sup_{n \geq 1} \int_a^b \| \Phi_n \|_F(t) \, dt <  \infty \, .
\ee
Then for any $X \in \mathcal{P}_0( \Gamma)$,  
\be
\lim_{n \to \infty}\Vert \tau_{t,s}^{\Phi_n}(A) -  \tau^{\Phi}_{t,s}(A)\Vert =0
\ee
for all $A \in \mathcal{A}_X$ and each $s,t \in I$. Moreover, the convergence is uniform for $s,t$ in compact intervals.
\end{theorem}

Now, consider a sequence $\Lambda_n=[a_n,b_n]\subset \Ir$, $\partial_n=(D_{1,n} , D_{2,n},s_n)$.
As the result of applying Proposition \ref{prop:loc_conv_in_fnorm} and Theorem \ref{thm:exist_iv_dyn}, we obtain the strong convergence of the finite-volume 
spectral flow automorphism generated by $\Psi^{\Lambda_n,\partial_n}$ to one and the same spectral flow for the infinite chain: for $\epsilon \in [0,1]$,
\be\label{sameTL}
\lim_{n\to\infty} \alpha_\epsilon^{\Lambda_n,\partial_n}(A) = \alpha_{\epsilon}(A), \quad \mbox{ for all } A\in\cA^{\rm loc}_\Gamma.
\ee

\section*{Acknowledgements}

Based upon work supported by the National Science Foundation under Grants and DMS-1207995 (AM), DMS-1515850 (AM and BN) and DMS-1813149 (BN).
We would like to thank the referee for helpful comments and suggestions.

\providecommand{\bysame}{\leavevmode\hbox to3em{\hrulefill}\thinspace}
\providecommand{\MR}{\relax\ifhmode\unskip\space\fi MR }
\providecommand{\MRhref}[2]{%
  \href{http://www.ams.org/mathscinet-getitem?mr=#1}{#2}
}
\providecommand{\href}[2]{#2}

\end{document}